\documentclass[12pt,onecolumn, draftclsnofoot]{IEEEtran}
\usepackage{url}
\usepackage{array}
\usepackage{makecell}
\usepackage[usenames,dvipsnames]{pstricks}
\usepackage{epsfig}
\usepackage{pst-grad} 
\usepackage{pst-plot} 
\usepackage{amsmath,epsfig}
\usepackage{amssymb}
\usepackage{mathtools}

\usepackage{enumitem}
\usepackage{psfrag}
\usepackage{color}
\usepackage{pstricks}
\usepackage{cite}
\usepackage{algorithm}
\usepackage{epsfig}
\usepackage{amsthm}
\usepackage{amssymb}
\usepackage{psfrag}
\usepackage{pstricks}
\usepackage{float}
\usepackage{stfloats}
\usepackage{wrapfig}
\usepackage{graphicx}
\usepackage{amsmath}
\usepackage{amsfonts}
\usepackage{amssymb}
\usepackage{algorithm}
\usepackage{blindtext}
\usepackage{pstricks,pst-plot}
\usepackage{pst-bezier}
\usepackage{pst-math}
\usepackage{algorithm,algcompatible,amsmath}
\usepackage{hyperref}
\usepackage{float}
\usepackage{cite}
\usepackage{amsmath,amssymb,amsfonts,dsfont}
\usepackage{scalerel,stackengine}
\stackMath
\newcommand\reallywidehat[1]{%
\savestack{\tmpbox}{\stretchto{%
  \scaleto{%
    \scalerel*[\widthof{\ensuremath{#1}}]{\kern-.6pt\bigwedge\kern-.6pt}%
    {\rule[-\textheight/2]{1ex}{\textheight}}
  }{\textheight}%
}{0.5ex}}%
\stackon[1pt]{#1}{\tmpbox}%
}
\usepackage{algorithm}
\usepackage[english]{babel}
\usepackage[utf8]{inputenc}
\usepackage[noend]{algpseudocode}
\usepackage{graphicx}
\usepackage{textcomp}
\usepackage{setspace}
\usepackage [english]{babel}
\usepackage [autostyle, english = american]{csquotes}
\usepackage[dvipsnames]{xcolor}
\colorlet{LightRubineRed}{RubineRed!70!}
\colorlet{Mycolor1}{green!10!orange!90!}
\definecolor{Mycolor2}{HTML}{000000}
\definecolor{Mycolor3}{HTML}{000000} 
\usepackage{tikz}

\algnewcommand\REQUIRED{\item[\textbf{Required:}]}%
\algnewcommand\INPUT{\item[\textbf{Input:}]}%
\algnewcommand\OUTPUT{\item[\textbf{Output:}]}%

\def\bg{{\bf g}}

\def\bc{{\bf c}}

\def\bo{{\bf o}}

\def\bm{{\bf m}}

\def\bs{{\bf s}}

\def\bx{{\bf x}}

\def\rs{{\mathsf{s}}}

\def\ro{{\mathsf{o}}}

\def\rx{{\mathsf{x}}}
\def\rc{{\mathsf{c}}}
\def\rm{{\mathsf{m}}}
\def\rtr{{\mathsf{tr}}}

        \newtheorem{theorem}{Theorem}

        \newtheorem{definition}[theorem]{Definition}

        \newtheorem{lemma}[theorem]{Lemma}

\title{\LARGE
Task-Oriented Communication Design at Scale \\ 
\thanks{The authors are with the Centre for Security Reliability and Trust, University of Luxembourg, Luxembourg. Emails: \{arsham.mostaani, thang.vu,  hamed.habibi, symeon.chatzinotas, bjorn.ottersten\}@uni.lu}
\thanks{This work is supported by European Research Council (ERC) advanced grant 2022 (Grant agreement ID: 742648).}}
 \author{\IEEEauthorblockN{Arsham Mostaani, \IEEEmembership{Student Member,~IEEE},
 Thang X. Vu, \IEEEmembership{Senior Member,~IEEE}, \\ Hamed Habibi, \IEEEmembership{Member,~IEEE},
 Symeon Chatzinotas, \IEEEmembership{Fellow Member,~IEEE}, \\ and Bj\"orn Ottersten, \IEEEmembership{Fellow Member,~IEEE} }
}
\begin{document}
\maketitle

\vspace{-18mm}
\begin{abstract}
\vspace{-4mm}
With countless promising applications in various domains such as IoT and industry 4.0, task-oriented communication design (TOCD) is getting accelerated attention from the research community. This paper presents a novel approach for designing scalable task-oriented quantization and communications in cooperative multi-agent systems (MAS). The proposed approach utilizes the TOCD framework and the value of information (VoI) concept to enable efficient communication of quantized observations among agents while maximizing the average return performance of the MAS, a parameter that quantifies the MAS's task effectiveness. The computational complexity of learning the VoI, however, grows exponentially with the number of agents. Thus, we propose a three-step framework: i) learning the VoI (using reinforcement learning (RL)) for a two-agent system, ii) designing the quantization policy for an $N$-agent MAS using the learned VoI for a range of bit-budgets and, (iii) learning the agents' control policies using RL while following the designed quantization policies in the earlier step. 
Our analytical results show the applicability of the proposed framework under a wide range of problems. Numerical results show striking improvements in reducing the computational complexity of obtaining VoI needed for the TOCD in a MAS problem without compromising the average return performance of the MAS.
\end{abstract}

\vspace{-10mm}

\vspace{-3mm}
\section{Introduction} \label{sec: Introduction}
\vspace{-3mm}


Be it the communication of data between distinct agents, or the transmission of information/signals inside a neural network (NN), communication and information exchange have always been an inseparable part of every data-driven learning system. For decades, the role of communication inside an AI system has been less investigated; often resulting in AI systems where communications are assumed to be perfect, e.g., the perfect communication of signals inside a NN. Nevertheless, communication is an integral part of AI, directly influencing its efficiency and accuracy. This is especially true when we view communications with its modern definitions steaming from the concept of task-oriented communications. In particular, with the rise of task-oriented communications \cite{mostaani2022survey,gunduz2022beyond}, there is a wide consensus about the diverse value of every bit sequence for a specific task \cite{mostaani2022task,soleymani2021value,howard1966information}. When the receiving end of communications intends to carry out a learning task, some bits in a sequence of the received communications might prove more useful. Under these circumstances, an effective design of communications can (i) significantly reduce the complexity of computations at the receiving end, (ii) improve the accuracy of the learning task under limited computational resources or equivalently, improve response time/latency of the receiving end in carrying out its computations, (iii) reduce the power consumed for communications and, (iv) reduce the required communication rate to overcome channel rate-constraints or network resource limitations \cite{mostaani2022task,liu2020edgefed}. While most of these benefits have a direct effect on the complexity as well as the accuracy of the AI system operating at the receiving end, they are also considered to be different aspects of the mission that task-oriented communications considers for itself - making it an integral part of AI.

Rethinking communications by understanding the value of communication bits can result in fundamental changes in the building blocks of (distributed) machine learning. While the literature of communications research has numerous examples of how machine learning can be leveraged to solve various communication problems \cite{8755300}, the contribution of communications in optimizing the (distributed) learning systems, which happens to be the main focus of this manuscript, has only recently started to receive attention from the research community \cite{ lin2016fixed,yang2020energy,jacob2018quantization,mostaani2022survey,tung2021effective,mota2021emergence,shlezinger2021deep,gutierrez2022learning,zhang2022goal,shlezinger2020task,mostaani2019Learning,mostaani2022task}.
In \cite{lin2016fixed,yang2020energy}, the authors investigate the effect of quantization channels for the transmission of signals from one neuron to another, inside different types of NNs. Federated learning over rate-reduced communication channels is investigated by \cite{jacob2018quantization}, resulting in reduced energy consumption for the whole distributed learning system. Direct task-oriented data quantization for an estimation task is introduced in \cite{shlezinger2020task}. Direct task-oriented communications for a user scheduling task is introduced by \cite{mota2021emergence}, achieving superior goodput performance. Indirect \footnote{By an indirect algorithm here we mean an approach that is not dependent on our knowledge of a particular task. Indirect approaches are applicable to any/(wide range of) tasks. In contrast to indirect schemes, we have direct schemes that are specifically designed for a niche application  \cite{shlezinger2020task}. As defined by \cite{mostaani2022survey}: "the direct schemes aim at guaranteeing or improving the performance of the cyber-physical system at a particular task by designing a task-tailored communication strategy".} design of communications for a classification task is carried out by \cite{shao2023edgeinference,xu2022classification} and for a variety of control tasks by \cite{mostaani2019Learning, tung2021effective, mostaani2022task,mostaani2022centralized}, with \cite{tung2021effective,mostaani2022centralized} being focused on star typologies for the communication network of the agents and \cite{mostaani2019Learning,mostaani2022task} on full mesh networks - Fig. \ref{fig: communication topology}.

 \begin{figure}[!t] 
  \centering 
      \includegraphics[width=0.6\textwidth]{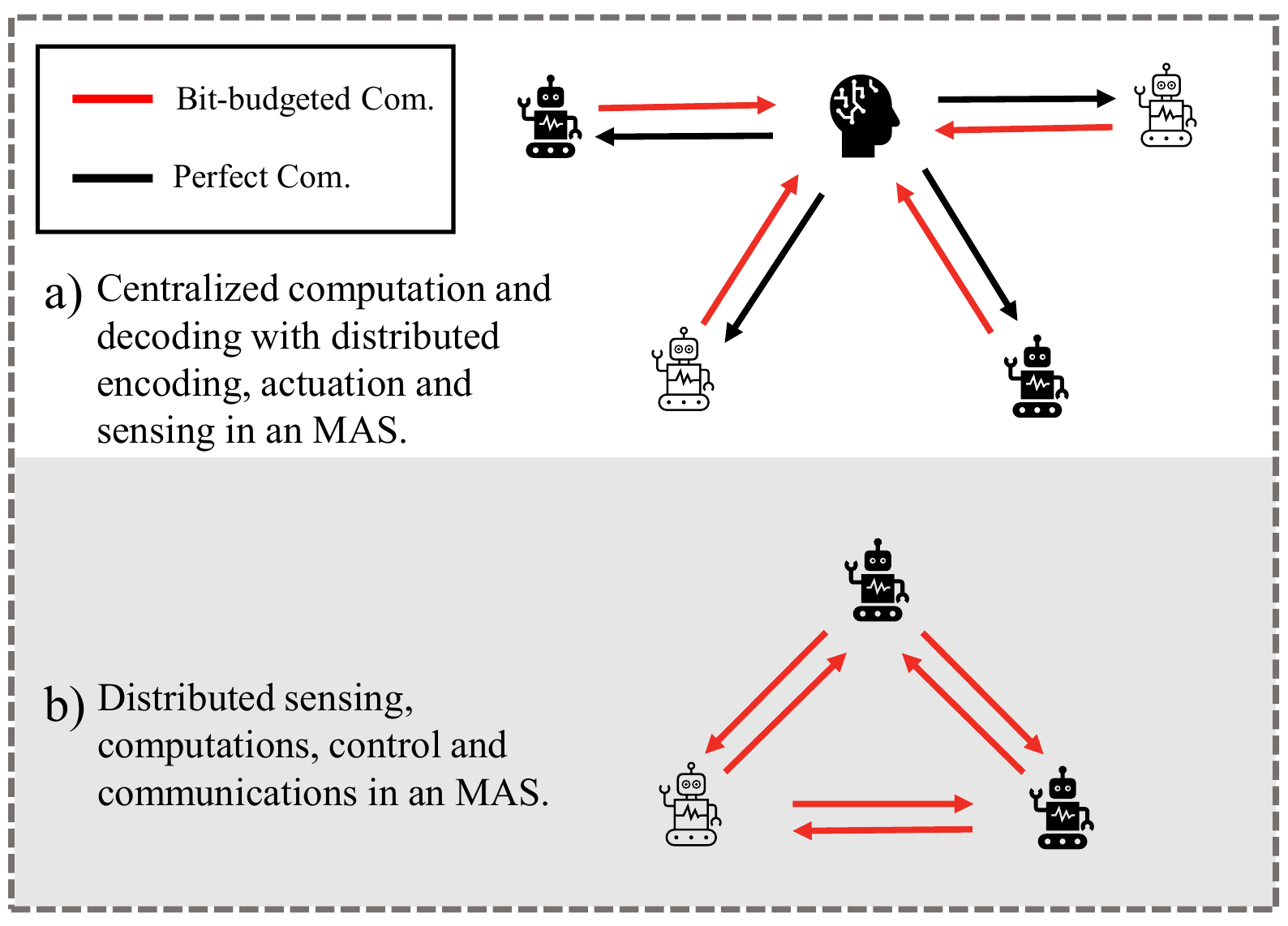} 
      \vspace{-4mm}
  \caption{ The communication network topology assumed in \cite{tung2021effective} vs. the adopted communication network topology in the current paper and in \cite{mostaani2022task}.}
  \label{fig: communication topology}
  \vspace{-0.5cm}
\end{figure}

By introducing the Dec-POMDPs \cite{oliehoek2008optimal}, recently, there has been a shift towards the joint design of communications and control \cite{tung2021effective}. However, we believe there is, yet, a huge potential in the disentanglement of the two problems, resulting in the introduction of task-effective communication design problems \cite{mostaani2022task,stavrou2022rate}\footnote{According to \cite{mostaani2022task}, task-effective communication problem is different from a traditional communication problem in that it captures some important features of the control task. }. In particular, supported by theoretical reasons, we decompose the joint problem into three separate but interdependent problems.  Decomposition of the two problems has a multitude of advantages: (i) it drastically reduces the complexity of the original joint problem - from a double exponential time problem \cite{bernstein2002complexity} to two exponential time and one quadratic time problems; (ii) it allows evaluating of the performance of our control and communication solutions in isolation \cite{mostaani2022centralized,lowe2019pitfalls}; (iii) it allows formulating a larger set of problems - in which agents can also communicate in an instantaneous fashion \cite{pynadath2002communicative} - as in the joint problem a delay in communications is inevitable \cite{oliehoek2008optimal}; (iv) solving the joint problem is oblivious to the inefficiencies of the communication solution, as we ultimately measure the effectiveness of the whole system according to the average task's cost/reward obtained by the joint communication and control solution. As per \cite{lowe2019pitfalls}, we can obtain a desirable performance in the task while the communications are not effective yet. Further, as \cite{mostaani2022centralized} suggests, and is shown in Fig. \ref{fig: motivating figure - separation}, the achievable average return of the system can be improved by increasing the memory of the receiving end's controller. However, it leads to much higher complexity for the controller to select suitable actions.
In fact, in the Dec-POMDP framework, obtaining effective distributed joint communication and control policies relies on processing the history of observations \cite{gronauer2022multi,oliehoek2008optimal}, with the complexity of the distributed policies growing exponentially\footnote{ "A DEC-POMDP
together with a joint policy can be viewed as a POMDP together with a policy, where the observations in the POMDP
correspond to the observation tuples in the DEC-POMDP.
In exponential time, each of the exponentially many possible sequences of observations can be converted into belief
states. The transition probabilities and expected rewards
for the corresponding "belief MDP" can be computed in
exponential time (Kaelbling et al., 1998)" \cite{bernstein2002complexity}.} with respect to the size of observation histories. Therefore, Dec-POMDP approaches can result in near-optimal control policies at the cost of expensive computational complexities \cite{gronauer2022multi}. The high cost of computations stems from the ineffectiveness of inter-agent communications that makes the decision-making more dependent on the larger history of observations.

 \begin{figure}[t!] 
  \centering 
      \includegraphics[width=0.6 \textwidth]{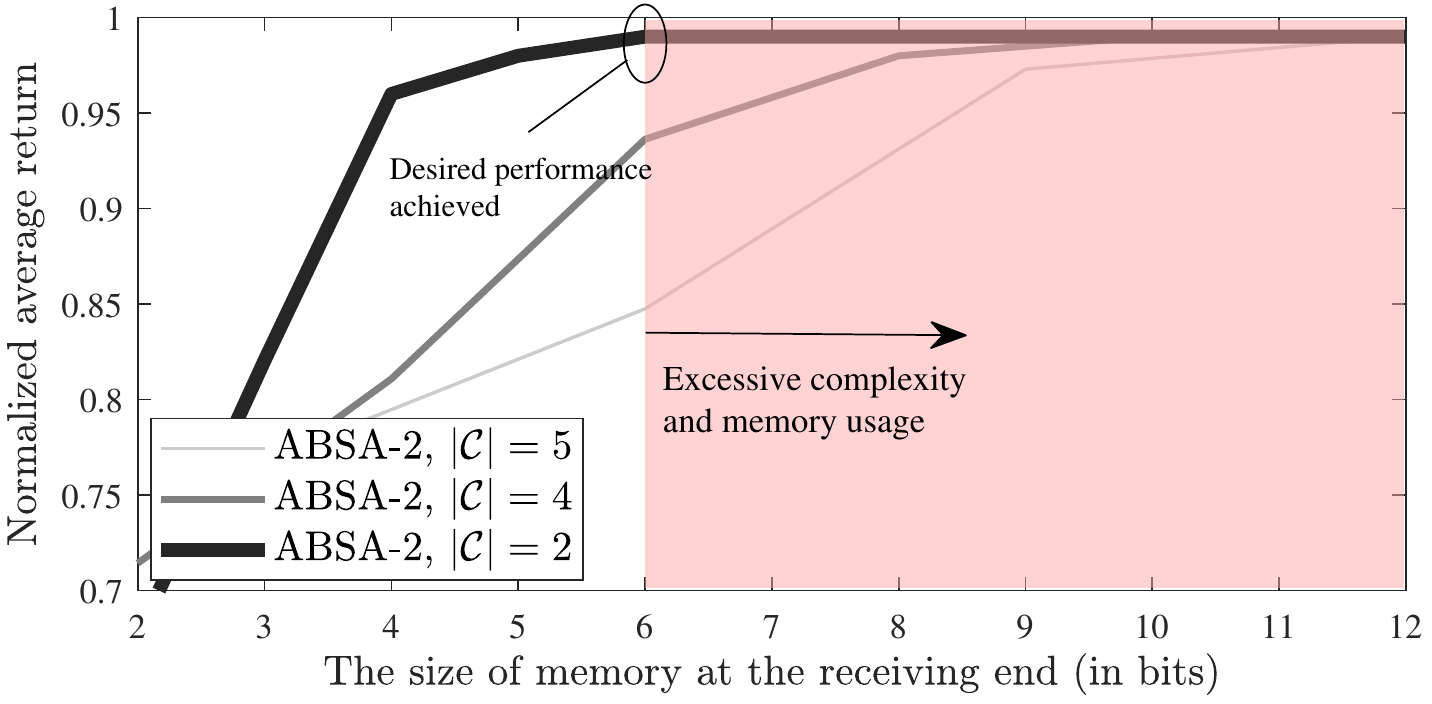} 
      \vspace{-4mm}
  \caption{Joint design of communications and control can potentially lead to inefficient communication policies whose weakness is compensated in the controller at the cost of radical increase in the complexity its running algorithms. The three curves shown in the figure, demonstrate the performance of Action-Based State Aggregation (ABSA) introduced in \cite{mostaani2022centralized}, at three different sizes of the quantization codebook $|\mathcal{C}|$. When the controller does not have access to the state information, regardless of the method used to design the communications to it, by increasing the memory of the controller, we can increase the average return performance of the system. Although the desired performance can be achieved by increasing the size of memory at the receiving end, this comes at the cost of a significant increase in the complexity of decision making at the receiver.}
  \label{fig: motivating figure - separation}
  \vspace{-0.4cm}
\end{figure}

The paper \cite{mostaani2022task}, is one of the first recent efforts to separate the data quantization and control policies. In contrast with the classic quantization problems \cite{1056142} where the goal is to minimize the distortion between the original signal and its quantized version, in the task-oriented/semantic quantization problem, the goal is to minimize the distortion between the task-relevant/semantic information available in the original signal v.s. the task-relevant/semantic information available in quantized signal \cite{stavrou2022rate,mostaani2022task}. The analysis provided in \cite{stavrou2022rate} and similar works \cite{liu2022indirect,poor21semantic} are not specific to a certain function that can capture the semantic/task-relevant data inside a given signal, whereas in \cite{mostaani2022task} the authors introduce a particular and indirect measure that can evaluate usefulness/relevance of an observation data for control tasks. The introduced measure, being referred to as value function in dynamic programming and reinforcement learning, is shown to be able to measure the importance of local observation data for a generic multi-agent control task over Markov decision processes (MDP)s \cite{mostaani2022task}. The complexity of computing the value function, however, is multiplied by the size of action-observation space with the addition of every agent to the system - making the value function extremely expensive to compute for multi-agent systems (MAS)s with large number of agents \cite{azar2011speedy,mostaani2022task}\footnote{Since the computational complexity of Q-learning in the centralized training phase is order of $|\Omega^n \times \mathcal{M}^n|$ time complexity \cite{azar2011speedy}, the addition of one single agent will multiply the complexity of the centralized training by $|\Omega \times \mathcal{M}|$.}. In this direction, the authors in \cite{gronauer2022rlsurvey} pronounce that: "since the complexity in the state and action space grows exponentially with the number of agents, even modern deep learning approaches may reach their limits.". 

Due to the importance of the scalability of multi-agent reinforcement learning (MARL) algorithms, the initial efforts to address this issue can be traced back to 1999 \cite{schneider1999distributed,wolpert1999general}, where authors introduce reward/value sharing for local optimizations. Since then, there has been a sustained effort to address the problem by other means such as introducing factored MDPs \cite{gronauer2022rlsurvey} or independent Q-learning \cite{tampuu2017multiagent}. Although the independent Q-learning and similar schemes \cite{tampuu2017multiagent} can scale with the growing number of agents they suffer from sub-optimality caused by the non-stationarity of the environment from each agent's perspective. This issue is addressed by modern MARL algorithms that comprise a centralized training phase, through which the training environment is guaranteed to stay stationary \cite{FoersterCounter}. The complexity of the centralized training phase, however, grows exponentially as the number of agents increases in the MAS. Monotonic Value Function Factorization - QMIX \cite{foerster2020qmix} - enforces a monotonicity constraint on the relationship between the local Q-functions and the centralized Q-function to reduce the complexity of factorizing the value decomposition networks. The overall complexity of MARL, however, increases with the addition of each agent to the system. Even with the use of attention mechanisms for centralized training, it has not been possible to go beyond linearly increasing the complexity of centralized training with respect to the number of agents \cite{iqbal2019actor}.

To the best of our knowledge, the current paper is the first to reduce the complexity of the centralized training phase from exponential time complexity - $\mathcal{O}( c^N)$ with $c>1$ - to constant time complexity - $\mathcal{O}(1)$ - with respect to the number of agents $N$. The only caveat is that our proposed scheme for reducing the complexity of the centralized training phase cannot be applied to every multi-agent learning problem but only to design the inter-agent communications in the multi-agent setting. In particular, the contributions of the paper are as follows.

\begin{itemize}
    \item We provide analytical studies to show that a two-agent centralized training phase is sufficient to draw the insights we need from the centralized training phase - when the reward function and observation structure follow certain conditions. Regardless of the method used in the centralized training to compute the value of observation space e.g., deep reinforcement learning, exact reinforcement learning or dynamic programming, our analytical results stay valid.
    
    \item The proposed analytical studies suggest that the value function obtained from the two-agent centralized training phase is sufficient to cast the task-oriented data quantization problem - even if we do not know the relationship between value function of the two-agent system versus $N$-agent system, $N$ being the real number of agents the system is composed of.
    
    \item According to these results, we propose a scalable state aggregation algorithm for data compression (ESAIC) which can easily be applied to MASs composed of a large number of agents with the aim fulfilling a collaborative task.
    
    \item By carrying out numerical studies on geometrical consensus problem \cite{barel2017come}, will show that the proposed ESAIC is capable of reducing the complexity of the centralized training for hundreds of days - if not years - even in very simple problems, while it maintains the average return performance of the algorithm close to the optimality.
\end{itemize}

\subsection{Organization}
Section II describes the system model for a cooperative multi-agent task with rate-constrained inter-agent communications. Section III provides a quick overview to SAIC, an exiting algorithm that can solve provide a solution to the joint control and data compression policy design problem. Our goal is to make SAIC computationally less complex in this manuscript. Section IV proposes the extended SAIC (ESAIC), a scheme for the joint design of data compression and control policies which is much less complex to run and very similar to SAIC in average return performance. We also provide analytical results on the conditions that ESAIC can maintain the performance of its predecessor. The numerical results and discussions are provided in section V. Finally, section VI concludes the paper.

\begin{table}
\caption{Table of notations}
\vspace{-0.3cm}
\centering
 \begin{tabular}{||c c ||} 
 \hline
 Symbol & Meaning \\ [0.5ex] 
 \hline\hline
 \small{$\bx(t)$} & \small{A generic random variable generated at time $t$}  \\ 
 \hline
 \small{$\rx(t)$} & \small{Realization of $\bx(t)$}  \\
 \hline
 \small{$\mathcal{X}$} & \small{Alphabet of \bx(t)}  \\
 \hline
 \small{$|\mathcal{X}|$} & \small{Cardinality of $\mathcal{X}$}  \\
 \hline
  \small{$\mathbb{P}(\mathcal{X})$} & \small{ Power set of $\mathcal{X}$ }  \\
 \hline
 \small{$p_{\bx}\big(\rx(t)\big)$} & \small{Shorthand for $\mathrm{Pr}\big(\bx(t) = \rx(t) \big)$}  \\  
 \hline
 \small{$H\big(\bx(t)\big)$} & \small{Information entropy of $\bx(t) $ (bits)}  \\  
 \hline
  \small{$\mathcal{X}_{-\bx}$} & \small{ $\mathcal{X} - \{\bx\}$}  \\ [1ex]
  \hline 
   \small{$\mathbb{E}_{p(\rx)}\{\bx\}$} & \small{\makecell{Expectation of the random variable $X$ over the \\ probability distribution $p(\rx)$}}  \\ [1ex]
 \hline
    \small{$\delta(\cdot)$} & \small{Dirac delta function}  \\ [1ex]
 \hline
   \small{$\rtr(t) $} & \small{Realization of the system's trajectory at time $t$}  \\ [1ex]
 \hline
\end{tabular}
\label{table-notation}
\vspace{-6mm}
\end{table}
We also use the concept of image functions in our analytical studies which is defined as the following. Let $g(\cdot) : \mathcal{D} \rightarrow \mathcal{C}$ be a function and $\mathcal{D}' \subset \mathcal{D}$ be a subset of its domain. The image function of $g(\cdot)$ denoted by $\Ddot{g}(\cdot) : \mathbb{P}(\mathcal{D}) \rightarrow \mathbb{P}(\mathcal{C})$ is defined as $ \Ddot{g}(\mathcal{D}') \triangleq \{ c \in \mathcal{C} \; | \; g(d) = c \; , d\in \mathcal{D}' \}$.
For the sake of the simplicity of the analysis, the arguments of the function may be omitted when no confusion is raised, e.g., we have used $r^{[n]}(\cdot)$ instead of $r^{[n]}(\ro_1,...,\ro_n,\rm_1,...,\rm_n)$. The focus of this paper is on the design of the quantizer of a communication pipeline; accordingly, the terms communication design and quantization design are used interchangeably throughout the manuscript.

\vspace{-4mm}
 \section{Problem Statement} \label{System model - Section}
 \vspace{-2mm}
We consider a multiagent system (MAS) in which  multiple agents $i \in \mathcal{N} = \{1,2,..., N\} $ collaboratively and distributedly execute a task. The system runs on discrete time steps $t$. The observation of each agent $i$ at time step $t$ is shown by $\bo_i(t)\in \Omega$ and the state $\bs(t) \in \mathcal{S}$ of the system is defined by the vector of joint observations $\bs(t) \triangleq [ \bo_i(t)]_{i \in \mathcal{N}} \in \Omega^N $. Now let $\bs_i(t) \in \{\Omega \cup {0}\} ^N $ be the vector of agent $i$'s local state, with all its elements being equal to zero except for its $i$'th element which is equal to $\bo_i(t)$. We assume that $\forall i,j \in \mathcal{N} \ $ the local states $\bs_i(t)$ and $\bs_j(t)$ are linearly independent. This is also referred to as joint observability of the state \cite{pynadath2002communicative}. The control action of each agent $i$ at the time $t$ is shown by $\bm_i (t) \in \mathcal{M}$, and the action vector $\bm(t) \in \mathcal{M}^N $ of the MAS is defined by the joint actions $\bm(t) \triangleq \langle \bm_1(t), ..., \bm_N(t)   \rangle $. The observation space $\Omega$, state-space $\mathcal{S}$, and action space $\mathcal{M}$ are all discrete sets.
%
The environment is governed by an underlying Markov Decision Process that is described by the tuple $M = \big{\langle} \mathcal{S},\mathcal{M}^N, r(\cdot), \gamma, T(\cdot)  \big{\rangle}$, where $r(\cdot): \mathcal{S} \times \mathcal{M}^N \rightarrow \mathbb{R}$ is the per-stage reward function and the scalar $ 0\leq\gamma \leq 1 $ is the discount factor. Also, the function $r^{[n]}(\cdot): \Omega^n \rightarrow \mathcal{M}^n$ is the reward function of an MAS comprised of $n$ agents. 
The function $T(\cdot): \mathcal{S}\times \mathcal{M}^N \times \mathcal{S} \rightarrow [0,1]$ is a conditional probability mass function (PMF) which represents state transitions such that $T\big(\rs(t+1), \rm(t), \rs(t)  \big) = \mathrm{Pr}\big( \rs(t+1) | \rs(t) , \rm(t) \big)$. The performance of the MAS is measured according to the system's average return defined as the summation of obtained per-stage rewards within the time horizon $T'$:
{\small\begin{equation} \label{eq: cumulative rewards}
    \bg(t^{'})= {\sum}_{t=t^{'}}^{T'}\gamma^{t-1} r\big(\bs(t),\bm(t)\big).
\end{equation}}
 Once per time step, following the Fig. \ref{fig: Comms model}, agent $i \in \mathcal{N}$ is allowed to transmit a communication vector $\rc_{i}(t)$ to every other agent $j \in \mathcal{N}_{-i} = \mathcal{N}_{-i}$ - following a full mesh topology for connectivity. Conditioned on its observation $\ro_i(t)$, agent $i$ transmits a vector of communication messages $\rc_i(t) = [  \rc_{i,j}(t) ]_{j \in \mathcal{N}_{-i}} \in \prod_{j\in \mathcal{N}_{-i}} \mathcal{C}_{i,j} $, in which the element $\rc_{i,j}(t)$ denotes the message sent by agent $i$ to agent $j$, where $\rc_{i,j}(t)$ is generated following the communication policy $\pi^c_{i,j}(\cdot) : \Omega \rightarrow \mathcal{C}_{i,j}$. The non-empty set $\mathcal{C}_{i,j}$ is an alphabet $\{ \rc_{i,j}, \rc_{i,j}', \rc_{i,j}'', ..., \rc_{i,j}^{(B_{i,j}-1)} \}$ composed of a finite $B_{i,j}$ number of communication code-words - we use the same notation to refer to the different elements of the action, observation and state spaces too. Agent $i$'s communications are generated by following the tuple $\pi^c_i = \langle  \pi^c_{i,j}(\cdot) \rangle_{j \in \mathcal{N}_{-i} } $ which is comprised of $N-1$ different communication policies. Agent $i$'s communications are sent over $N-1$ separate error-free finite-rate bit pipe, with its rate constraint to be $R_{i,j} \in \mathbb{N}$ (bits per channel use) or equivalently (bits per time step) \footnote{The finite rate of the bit pipe - or equivalently the bit-budget $R_{i,j}$ - is used as the metric to create different peer to peer (P2P) scenarios for the communication between any two agents $i, j \in \mathcal{N}$. We can use the same bit-budget value to represent a P2P system that runs under a higher power/higher noise as long as the SNR and the linear code word generator matrix (of the channel coding) are the same see e.g., \cite{viterbi1971convolutional} for more detailed discussions. In fact, the considered bit-budget is equivalent to the achievable rate of the operating modulation and (channel) coding.
}. As a result, the cardinality of the communication symbol space $\mathcal{C}_{i,j}$ for each $i$ to $j$ inter-agent communication link should follow the inequality
{\small \begin{align} \label{eq: bit-budget constraint}
    0 \leq B_{i,j} \leq 2^{R_{i,j}}.
\end{align}}
 In the special case of homogeneous bit-budgets, we have $R_{i,j} = R, \forall i,j \in \mathcal{N}$. Each agent $i$ exploits its observation $\bo_i(t)$ together with the received communication messages $ \tilde{\rc}_i(t) = [  \rc_{j,i}(t) ]_{j \in \mathcal{N}_{-i} } \in \prod_{j\in \mathcal{N}_{-i}} \mathcal{C}_{j,i} $ within time-step $t$ to select the control signal $\rm_i(t)$ following a deterministic control policy $\pi^m_i(\cdot): \prod_{j\in \mathcal{N}_{-i}} \mathcal{C}_{j,i} \times \Omega \rightarrow \mathcal{M}$. Accordingly, the problem we solve is detailed in Definition 1.


 \begin{figure}[!t] 
  \centering 
      \includegraphics[width=0.6\textwidth]{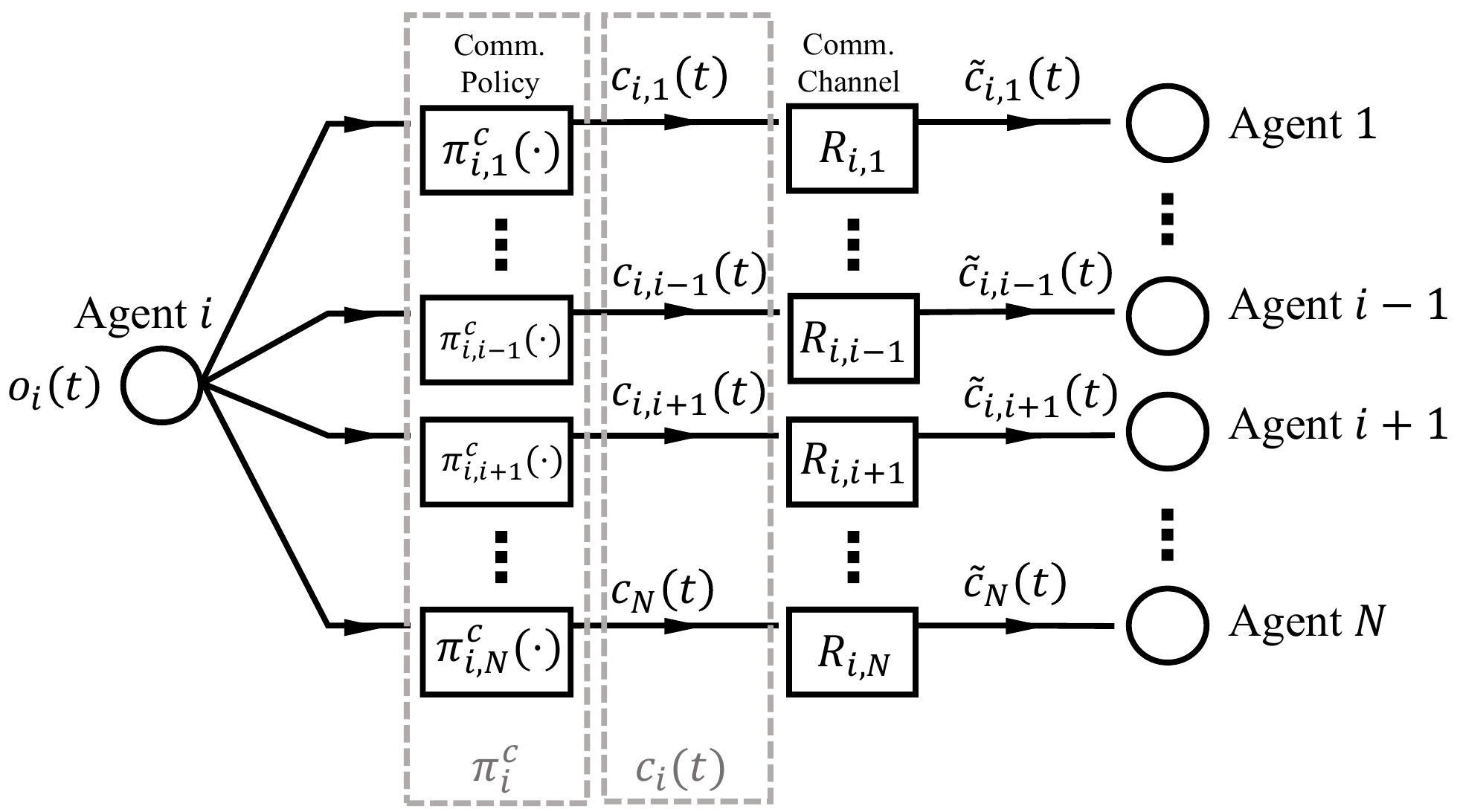} 
      \vspace{-2mm}
  \caption{Illustration of message transmission (encoding) at agents $i$. Agent $i$'s observation $\ro_i(t)$ at time step $t$ is transmitted to all other agents. Each inter-agent communication channel from agent $i$ to agent $j$ is assumed to be reliable so long as bit-budgets requirements - explained in (\ref{eq: bit-budget constraint}) - are respected. }
  \label{fig: Comms model}
  \vspace{-0.2cm}
\end{figure}
%

\vspace{-0mm}
\begin{definition} \textbf{(Distributed Joint Control and Communication Design (D-JCCD) problem).}
 Let $M$ be the MDP governing the environment and the scalar $R_{i,j} \in \mathbb{R}$ to be the bit-budget of each inter-agent communication channels. At any time step $t'$, we aim at designing the tuple $\pi_i = \langle \pi^m_i(\cdot), \pi^c_{i} \rangle$ to solve the following variational dynamic programming
{\small \begin{align}\label{eq: Decentralized Joint Control and Communication Design (D-JCCD) problem}
 \underset{\pi_i}{\textnormal{argmax}} ~~ 
\mathbb{E}_{\pi_i}                  \Big\{
   \bg(t') 
\Big\};~~
 \textnormal{s.t.} ~~  
    B_{i,j} \leq 2^{R_{i,j}},  \,\, \forall  i,j \in \mathcal{N}
\end{align} }
\sloppy where the expectation is taken over the joint pmf of system's trajectory $\{\rtr\}_{t'}^{T'} = [\ro_1(t'),..., \ro_N(t'), \rm(t'), ..., \ro_1(T'),..., \ro_N(T'), \rm(T')]$, when each agent $i$ follows the policy $\pi_i$ for all agents $i \in \mathcal{N}$.
\end{definition}

In contrast to \cite{mostaani2022task}, we do not characterize the performance gap caused by the limited connectivity in the communication network of agents. Characterizing the difference between the performance of the MAS that runs over heterogeneous bit-budgets and the MAS that runs over perfect communication channels is deferred to future works. The present paper, however, will provide numerical studies on the performance of the proposed scheme - ESAIC - under asymmetrical communication bit-budgets  $R_{i,j}$.


\section{State Aggregation for Information Compression (SAIC)}

Instead of directly solving the D-JCCD problem (\ref{eq: Decentralized Joint Control and Communication Design (D-JCCD) problem}), problem \eqref{eq: Decentralized Joint Control and Communication Design (D-JCCD) problem} is tackled via SAIC in a three-step process as proposed in \cite{mostaani2022task} - each described in a separate subsection as follows.
\vspace{-0.3cm}
\subsection{Centralized training phase}
\vspace{-0.2cm}
Following a centralized training and distributed execution approach \cite{FoersterLearning, FoersterCounter}, SAIC solves the problem from a centralized point of view where all the communications between agents and a central controller are considered to be perfect:

{\small
\begin{align}\label{centralized problem - general problem}
 \pi^*(\cdot) = \underset{\pi(\cdot)}{\text{argmax}}
 \,\,\,\mathbb{E}_{\pi}                  \!\Big{\{}
\bg(t)
                \!\Big{\}},  
\end{align}}
and the policy $\pi$ can be expressed as a CMF $
    \pi \Big( \rm(t) \Big{|} \rs(t)  \Big)  =  p \big(\rm(t) \big{|} \rs(t) \big)$. In the centralized problem (\ref{centralized problem - general problem}), the
  objective is to design one centralized control strategy $\pi(\cdot): \Omega^N \rightarrow \mathcal{M}^N$. 

  \vspace{-0.3cm}
  \subsection{Task-oriented communication/quantization design (TOCD)}
  \vspace{-0.2cm}
  The result obtained in the previous phase, centralized training, is transferred to the next phase via the means of the value function. That is, we use the optimal control policy $\pi^*(\cdot)$ to compute the value of observations - see the relationship between the two in \eqref{value function iterated expectation simplified - State aggregation}. The value function helps measure how valuable each agent's observations are for other agents' decision-making. The non-injective surjective mapping  $V^*(\cdot): \Omega \rightarrow \mathcal{V} \subset \mathbb{R} $, that is obtained after solving the centralized problem, allows us to solve the communication problem over the output space of the mapping $V^*(\cdot)$ - value space $\mathcal{V}$ - rather than over the original observation space.

It is shown in \cite{mostaani2020task,mostaani2020state} that, after obtaining $\pi^{*}(\cdot)$ as the optimal solution to (\ref{centralized problem - general problem}), one can obtain the value of each observation $\ro^{(k)}$ for all $k \in \{1, ..., |\Omega| \}$ following the
 {\small\begin{align} \label{value function iterated expectation simplified - State aggregation}
   & V^{* [N]}\big( \ro_i(t) = \ro^{(k)}\big) = \\ 
   &\sum_{{\ro_{-i}(t)  \in \Omega^{N-1}} }    
   \mathbb{E}_{\pi^*}                  \Big\{
   \bg(t)  \big{|}
   \bs(t) , \pi^*\big( \bs(t) \big)  
\Big\}
   p\big(\bo_{-i}(t) = \ro_{-i}(t)\big),       \notag
\end{align}}
where $\bs(t) = [ \ro_1(t), \dots, \ro_N
(t)]$ and the summation $\sum_{{\ro_{-i}(t)  \in \Omega^{N-1}} }$ is used to denote $N-1$ summations over all possible values for $ \ro_{-i}(t) =  [\ro_i(t)]_{i \in \mathcal{N}_{-i}}$. As also shown in the transition of Fig. \ref{fig: SAIC-Communications Design}. a to Fig. \ref{fig: SAIC-Communications Design}. b, by knowing the mapping $V^{* [N]}(\cdot)$ we can map all the observation values to the one-dimensional value space $\mathcal{V}$. Accordingly, the clustering of observation points will no longer be done based on their observation values e.g., $\ro_i(t)$, but based on the value function of the observation values, e.g., $V^{* [N]}(\ro_i(t))$ - where the superscript $[N]$ illustrates the number of agents in the centralized training phase \footnote{Note that, whenever a function/policy - e.g., $\pi^c_{i,j}(\cdot)$ - is obtained via a centralized training which has had $N'$ number of agents in it, the superscript $[N']$ is used for that function/policy - e.g., $\pi^{c,[N']}_{i,j}(\cdot)$.}. This would result in solving the following TODC problem in the form of a clustering problem
    \vspace{-2mm}
    {\small \begin{equation} \label{Task-Based Information Compression}
        \begin{aligned}
        &  \underset{\mathcal{P}_{i,j}}{\text{min}}
        & & {\sum}_{k=1}^{2^{R_{i,j}}} {\sum}_{\ro \in \mathcal{P}_{i,k}} \Big{|}                 V^{* [N]}\big(\ro_i(t) = \ro \big) - \mu^{'}_k \Big{|}, 
        \end{aligned}
    \end{equation}}
    to cluster observations via the partition $\mathcal{P}_{i,j} =  \{ \mathcal{P}_{i,j,1}, \dots, \mathcal{P}_{i,j,B_{i,j}} \}$ and learn the communication strategy $\pi^c_{i,j}(\cdot)$, where for each $\mathcal{P}_{i,j}$, all the observations $\ro_i(t) \in \mathcal{P}_{i,j,k}, \forall k \in \{1,..., B_{i,j}\}$ are corresponded to a single unique code-word $\pi^{c}_{i,j}(\ro_i(t))$, equivalently the output of the image function $ \Ddot{\pi}^{c}_{i,j}(\mathcal{P}_{i,j,k}), \forall k \in \{1,..., B_{i,j}\}$ is a single member set. Solving problem (\ref{Task-Based Information Compression}) is illustrated in Fig. \ref{fig: SAIC-Communications Design} by a transition from subplot "b" to "c". It was shown in \cite{mostaani2020state} that the optimal solution to (\ref{Task-Based Information Compression}) is the optimal solution to an approximated version of the problem (\ref{eq: Decentralized Joint Control and Communication Design (D-JCCD) problem}). Note that for every agent $i$, the problem (\ref{Task-Based Information Compression}), should be solved $N^c_i$ number of times where $N^c_i$ denotes the of agents $j \in \mathcal{N}_{-i}$ for whom the bit-budget of communications $B_{i,j}$ are distinct. Equivalently, the communication policy of agent $i$ for the two distinct receiving ends $j,j' \in \mathcal{N}_{-i}$ will stay the same if $B_{i,j} = B_{i,j'}$.
    \vspace{-0.3cm}
    \subsection{Decentralized Training Phase} \label{subsect: ESAIC: SAIC: decentralized training phase}
    \vspace{-0.2cm}
    Once the clustering problem (\ref{Task-Based Information Compression}) is solved, we have obtained indirect task-effective communication policies $\pi^c_i, \forall i \in \mathcal{N}$ \footnote{These quantization policies are indirect since they can be obtained without any prior knowledge about the task. And, they are task-effective, since they can be designed to preserve the accuracy of observation data when the observation data is deemed valuable for the specific task.}. To completely solve the D-JCCD problem (\ref{eq: Decentralized Joint Control and Communication Design (D-JCCD) problem}) via SAIC, we still have to find the optimal control policy $\pi^m_i(\cdot)$ for each agent $i$. Via the control policy $\pi^m_i(\cdot)$, at any time step $t$, agent $i$ selects a control signal $\bm_i(t)$, conditioned only on the quantized data received from the other agents $\tilde{\bc}_i(t) \in \prod_{j\in \mathcal{N}_{-i}} \mathcal{C}_{j,i} $, together with its own observation $\bo_i(t) \in \Omega$. SAIC obtains the control policy $\pi^m_i(\cdot)$ for each agent $i$, via a distributed training phase, in which agents communicate through bit-budgeted communication channels - following (\ref{eq: bit-budget constraint}). To this aim, the communications of each agent $i \in \mathcal{N}$ to each agent $j \in \mathcal{N}_{-i}$ are carried out via the communication policy $\pi^c_{i,j}(\cdot)$ that is obtained by solving (\ref{Task-Based Information Compression}). To obtain asymptotically optimal control policies, SAIC utilizes distributed Q-learning \cite{lauer2000distributedQ} for the distributed training phase.
    \vspace{-0.3cm}
    \subsection{Computational Complexity}
    \vspace{-0.2cm}
    The computational complexity of the centralized training phase $\mathcal{O}(|\Omega|^N \times |\mathcal{M}|^N)$ for a certain number of agents $N$ grows linearly with the size of observation and action spaces and for a certain size of observation-action space grows exponentially with the number of agents $N$. This makes computational cost of SAIC for large MASs prohibitively high, limiting its application to MASs composed of only a few agents. Given the exponential time complexity of the centralized training phase and its much higher time complexity compared with the distributed training phase, the centralized training phase is the major computational bottleneck of SAIC.

\vspace{-3mm}
\section{Extended State Aggregation for Information Compression in Multiagent Coordination Tasks}
\label{sec: ESAIC: idea and algorithm}
\vspace{-1mm}

 \begin{figure*}[!t] 
  \centering 
      \includegraphics[width=0.65\textwidth]{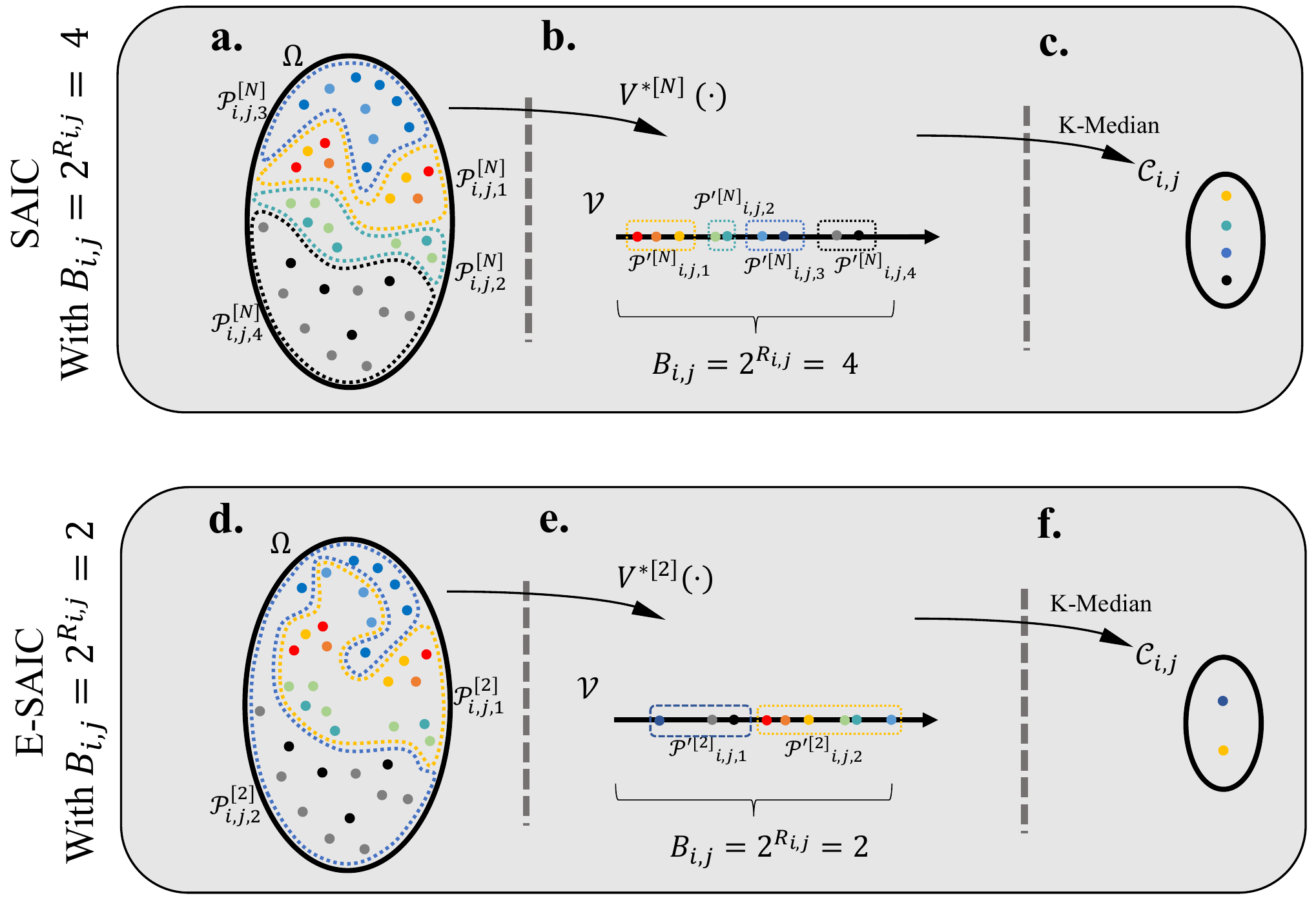} 
      \vspace{-4mm}
  \caption{Illustration of the steps taken to design the communication policy $\pi^c_{i,j}(\cdot)$ using SAIC and ESAIC.}
  \label{fig: SAIC-Communications Design}
  \vspace{-0.4cm}
\end{figure*}

In this section, we propose a straightforward extension of SAIC called Extended SAIC (ESAIC) which is capable of drastically reducing its time complexity in the centralized training phase. While the time complexity of the centralized training phase in SAIC grows exponentially with respect to the number of agents, in ESAIC, increasing the number of agents in the MAS, has no impact on the computational complexity of the centralized training phase - making ESAIC more efficient than SAIC \cite{mostaani2022task} and any other MARL with a central training phase \cite{FoersterCounter,foerster2020qmix,mostaani2022centralized}. ESAIC is not just a replacement for SAIC, but introduces the more general idea of reducing the number of agents in the MAS for the centralized training phase. Extended SAIC, proceeds by following the same steps as SAIC to solve the D-JCCD problem: (i) centralized training phase, (ii) task-oriented data compression problem, (iii) distributed training of agents' control policies. The only difference is that the centralized training phase is done with only two agents in the training phase - regardless of the number of agents $N$ for which we want to solve the original D-JCCD problem (\ref{eq: Decentralized Joint Control and Communication Design (D-JCCD) problem}).
\vspace{-0.3cm}
\subsection{Centralized Training Phase}
\vspace{-0.2cm}
Accordingly, by carrying out the centralized training phase, we solve the problem (\ref{centralized problem - two-agent problem}) for a two-agent system to obtain $V^{*[2]}(\cdot)$, $\pi^{*[2]}(\cdot)$ following
{\small
\begin{align}\label{centralized problem - two-agent problem}
 \pi^{*[2]}(\cdot) = \underset{\pi(\cdot)}{\text{argmax}}
 \,\,\,\mathbb{E}_{\pi}                  \!\Big{\{}
\bg(t)
                \!\Big{\}}. 
\end{align}}
Obtaining the value function $V^*(\cdot)$ is imperative both in SAIC and ESAIC because of a multitude of reasons: (i) The mapping $V^*(\cdot)$ projects the high-dimensional observation points to the single-dimensional space of $\mathcal{V} \subset \mathbb{R}$, leading to a reduced the complexity for the clustering problem, (ii) the mapping $V^*(\cdot)$ captures the features of the control task and allows us to take these features into account in our communication design problem - that helps us to separately design communications and control policies, (iii) the clusters in the output space of the $V^*(\cdot)$ are shown to be linearly separable, (iv) last but not least, it is very intuitive to see how the mapping $V^*(\cdot)$ is an indirect/universal measure to quantify the value of each observation for any given task. Accordingly, the observation points are not clustered together based on how similar they are, but based on how similarly valuable they are for the task.
\vspace{-0.4cm}
\subsection{Task-oriented communication/quantization design (TOCD)}
\vspace{-0.2cm}
Afterwards, by solving the following task-oriented quantization problem
\vspace{-2mm}
    {\small \begin{equation} \label{Task-Based Information Compression - ESAIC}
        \begin{aligned}
        &  \underset{\mathcal{P}^{[2]}_{i,j}}{\text{min}}
        & & {\sum}_{k=1}^{2^{R_{i,j}}} {\sum}_{\ro \in \mathcal{P}_{i,k}} \Big{|}                 V^{* [2]}\big(\ro_i(t) = \ro \big) - \mu^{'}_k \Big{|},
        \end{aligned}
    \end{equation}}
we obtain a new partition $\mathcal{P}^{[2]}_{i,j}$ of the observation space that leads to a different, yet effective communication/quantization policy $\pi^{c [2]}_{i,j}(\cdot)$. \textcolor{Mycolor3}{ K-median clustering can be used to solve the above-mentioned problem (\ref{Task-Based Information Compression - ESAIC}). In this direction, to obtain the quantization policy of agent $i$ for its communication to agent $j$ we compute a partition $\mathcal{P}'^{[2]}_{i,j}$ of the set $\mathcal{V}^{[2]}_{i}$ - where $\mathcal{V}^{[2]}_{i}$ is the image of $\Omega$ under the function $V^{* [2]}(\cdot)$ i.e., $\mathcal{V}^{[2]}_{i} = \ddot{V}^{* [2]}(\Omega)$. We first solve the following problem
\vspace{-2mm}
    {\small \begin{equation} \label{eq: ESAIC: TOCD}
        \begin{aligned}
        &  \underset{\mathcal{P}'^{[2]}_{i,j}}{\text{min}}
        & & {\sum}_{k=1}^{2^{R_{i,j}}} {\sum}_{V^*(\ro_i(t)) \in \mathcal\mathcal{P}'^{[2]}_{i,j,k}} \Big{|}                 V^{* [2]}\big(\ro_i(t)\big) - \mu^{''}_k \Big{|}.
        \end{aligned}
    \end{equation} }
 Afterwards, as shown in Figure \ref{fig: SAIC-Communications Design}, the observation points should be clustered according to the clustering of their corresponding values. That is, any two distinct observation points $\ro'_i, \ro''_i \in \Omega$ are clustered together in $\mathcal{P}_{i,j,k}$ if and only if their values $V^{* [2]}(\ro'_i), V^{* [2]}(\ro''_i) \in \mathcal{P}'_{i,j}$ are in the same cluster $\mathcal{P}'_{i,j,k}$.} Accordingly, each agent $i$ solves the problem \eqref{eq: ESAIC: TOCD} for $N_i \leq N$ number of times where, $N_i^c$ stands for the distinct number of bit-budgets at which agent $i$ has to communicate with other agents in the network. 
\vspace{-0.4cm}
\subsection{Decentralized Training Phase}
\vspace{-0.2cm}
After obtaining the communication policies, we solve the following distributed control design problem
\vspace{-4mm}
{\small\begin{align}\label{eq: Decentralized training problem}
 \underset{\pi^m_i}{\textnormal{argmax}} ~~ 
\mathbb{E}_{\pi_i}                  \Big\{
   \bg(t') 
\Big\}, ~~ \forall i \in \mathcal{N}
\end{align}}
through a distributed training phase to obtain the control policy of each agent $i$, where the expectation is taken over the MAS's trajectory that is influenced by both the control policy $\pi^c_i(\cdot)$ and the communication/quantization policy $\pi^m_i(\cdot)$ of all agents $i \in \mathcal{N}$. The detailed recipe of ESAIC can be found in Algorithm 1 and its performance will be studied both analytically and numerically in sections \ref{sec: analytical studies} and \ref{sec: numerical studies}, respectively.

As will be shown in section \ref{sec: analytical studies}, the number of agents in the training phase can be reduced, regardless of the specific method used to compute the function $V^{*[2]}(\cdot)$. Accordingly, we conjecture that other schemes such as deep Q-learning \cite{mnih2015human}, deep double Q-learning \cite{van2016deep}, deep deterministic policy gradient \cite{lillicrap2015continuous} and other similar (deep) reinforcement learning algorithms can be used for a two-agent centralized training phase to approximate the value function $V^{*[2]}(\cdot)$ - as long as the condition of theorem \ref{theorem: main theorem} is met.


\begin{algorithm}\label{alg: ESAIC - Algorithm}
\caption{ Extended State Aggregation for Information Compression (ESAIC)}
\begin{algorithmic}[1]

\State \small \textbf{Input:} $\gamma$, $\alpha$, $c$
 \State \textbf{Initialize} all-zero Q-table $Q^{m}_{i}(\cdot) \leftarrow Q^{m,(k-1)}_{i}(\cdot)$, for $i=1:N$
 \State $\;\;\;\;\;\;\;\;\;\;\;\;\;\;\!$ and all-zero Q-table $Q\big(\rs(t),\rm(t)\big)$.
 \State Obtain $\pi^{* [2]}(\cdot) \text{ \& } Q^{* [2]}(\cdot)$ by solving (\ref{centralized problem - general problem}) using Q-learning \cite{Suttonintroduction}. 
 \State Compute $V^{* [2]}\big( \ro_i(t) \big)$ following eq. (\ref{value function iterated expectation simplified - State aggregation}), for $\forall \ro_i(t) \in \Omega$.
 \State Obtain $\pi^{c [2]}_i$ by solving the problem (\ref{Task-Based Information Compression - ESAIC}) $N^c_i$ times, for $i=1:N$.
  \For{each episode $k=1:K$} 
 \State {\small Randomly initialize the observation $\ro_i(t=1)$, for $i=1:N$}
\For{$t_k = 1:M$}

\vspace{1mm}          
            \State Select $\rc_i(t)$ following $\pi^{c [2]}_i(\cdot)$, for $i=1:N$

            \State Obtain message $\tilde{\rc}_i(t)$, for $i=1:N $
            
            \State Update $Q^{m}_{i}\big(\ro_i(t-1),\tilde{\rc}_i(t-1),\rm_i(t-1)\big)$
            , for $i=1:N$
            
            \State Select $\rm_i(t) \in \mathcal{M}$ following $\epsilon$-greedy, for $i=1:N$
            \State Obtain reward $r\big( \rs(t),\rm(t) \big)$, {for} $i=1:N$
            
            \State Make a local observation $\ro_i(t)$, for $i=1:N$
    \vspace{1mm}

\State $t_k=t_k+1$

\EndFor\label{euclidendwhile-2}
\State \textbf{end}
\State Compute $\sum^{M}_{t=1} \gamma^{t-1} r_t$ for the $l$th episode
\State update $\epsilon$ via: $\epsilon =  -0.99 k/K + 1$
\EndFor
\State \textbf{end}
\vspace{1mm}

\State \textbf{Output:} $Q^{m}_{i}(\cdot)$ and  $\pi_{i}^{m}\big(\rm_i(t)|\ro_i(t),\tilde{\rc}_i(t)\big)$,
{for} $i=1:N$

%
\vspace{-1mm}
\end{algorithmic}
\end{algorithm}

\vspace{-0.4cm}
\section{Analytical study of ESAIC} \label{sec: analytical studies}
\vspace{-0.4cm}

After introducing the idea of ESAIC, in \ref{sec: ESAIC: idea and algorithm}, in this section, we provide analytical studies on its average return performance as well as studies on its computational complexity.

\subsection{Average return performance}
The main result of this subsection is to prove that by solving the problem (\ref{Task-Based Information Compression - ESAIC}), one can obtain inter-agent communication/quantization policies which are as effective as the solutions to the problem (\ref{Task-Based Information Compression}). Equivalently, one can reduce the number of agents in the centralized training phase and yet draw enough insights from it to design task-oriented communication policies. The proof provided in this section, therefore, is a testament to how rich is the value function of a two-agent centralized training phase to indirectly incorporate the features of the control task into the task-oriented communication design problem (\ref{Task-Based Information Compression - ESAIC}). These features have been previously extracted e.g., from the control problem through the Eigenvalues of the plant\footnote{In the terminology of reinforcement learning, the plant is referred to as the environment.} to be controlled \cite{tatikonda2004control} - for linear time-invariant plants.

\begin{theorem} \label{theorem: main theorem}
 Let the bijection $f(\cdot): \mathcal{V}^{[2]} \rightarrow \mathcal{V}^{[N]} $ be the mapping from the value of observations for a two-agent scenario to the $N$-agent. For all $i,j \in \mathcal{N}$, the partition $\mathcal{P}^{[2]}_{i,j}$ proposed by ESAIC (that is obtained by solving the problem (\ref{Task-Based Information Compression - ESAIC})) are the same as the partition $\mathcal{P}^{[N]}_{i,j}$ proposed by SAIC (that is obtained by solving the problem (\ref{Task-Based Information Compression})) if
 \vspace{-4mm}
 {\small \begin{align}
     &  c_1: \label{condition: main condition--}
     & & \forall k \in \{1,\dots, B_{i,j}\} \,\,\, \exists k' \in \{1,\dots, B_{i,j}\}  : \\
     & & & \Ddot{f}(\mathcal{P}'^{[2]}_{i,j,k}) = \mathcal{P}'^{[N]}_{i,j,k'}  \notag 
 \end{align}}
\end{theorem}

\begin{proof}
Appendix  \ref{app: proof: theorem: theorem: main theorem}.
\end{proof}

 

\textit{Remark 1:} Following the theorem \ref{theorem: main theorem}, all the guarantees that are presented for the performance of SAIC are in place if $R_{i,j} = R ~~ \forall i,j \in \mathcal{N}$.

Theorem \ref{theorem: main theorem}, provides a conditional guarantee for the equivalence of the results obtained by SAIC and ESAIC. The condition, however, is not always easy to verify. In the next section, we will provide numerical results also for the cases where the condition of the theorem is violated. The near-optimal performance of ESAIC, even under the violation of $c_1$ in (\ref{condition: main condition--}), confirms that the proposed condition is too strong and can be further relaxed in future works. Moreover, once we see certain structures and features in the function $f(\cdot)$ for smaller MASs, they may hold for larger MASs too. This would provide us with an analytical basis to use proof by induction to verify the condition $c_1$. The following remarks, introduce some of these features. 

\textit{Remark 2:} If the function $f(\cdot)$ is limit-preserving then it meets the condition $c_1$ \cite{davey_priestley_2002}.

\textit{Remark 3:} If the function $f(\cdot)$ is strictly monotonic, then it is limit-preserving too \cite{davey_priestley_2002}.

In particular, let the superscript in ${f}^{[N]}(\cdot)$ determine the superscript of its range $\mathcal{V}^{[N]}$. We will show in lemma \ref{lemma: verfiablity} that if $f^{[3]}(\cdot)$ is strictly monotonic, so is $f^{[N]}(\cdot)$ - for a specific class of reward functions and observation structures. Accordingly, to verify the condition $c_1$ for every $N \geq 3$, it will be sufficient to just verify it for $N=3$. 

\begin{lemma} \label{lemma: verfiablity}
    Let the function $f^{[3]}(\cdot): \mathcal{V}^{[2]} \rightarrow \mathcal{V}^{[3]}$ be strictly monotonic, and the conditions $c_2$ and $c_3$ met, defined as follows:

    $c_2$: The discrete derivative of the $r^{[k]}(\cdot)$ with respect to $k$ be a linear function, i.e., there exist scalars $\tau \in \mathbb{R}^+$ and $\zeta \in \mathbb{R}$ such that
\begin{equation} \label{eq: linear relation-lem}
    r^{[k+1]}(t) = \tau r^{[k]} + \zeta,
\end{equation}

$c_3$: The observations $\ro_{i}(t)$ of the $i$'th agent are independent of the observations $\ro_{j}(t)$ of the $j$'th agent $\forall i,j \in \mathcal{N}$ and $\forall \ro \in \Omega$.

Then, the Proposition $P(N)$ holds true as follows:

\textbf{$P(N)$:} The function $f^{[N]}(\cdot): \mathcal{V}^{[2]} \rightarrow \mathcal{V}^{[N]}$ is strictly monotonic for all $N\geq 3$.
\end{lemma}

\begin{proof}
    Appendix 
    B.
\end{proof}

\emph{Remark 4}: Quadratic cost functions which a popular class of reward/cost function used for optimal control - see,  e.g., \cite{kostina2019rate} - follow the condition $c_2$. The reward function used in this paper, which is not a quadratic function, is among many other functions \cite{mostaani2019Learning, zilber2001communication, tung2021effective} that meet this condition.

\emph{Remark 5}: Condition $c_3$ is commonly assumed in the multi-agent system's literature \cite{kim2006exploiting,varakantham2007letting}.
 Intuitively, we know that in remote collaborative applications where agents are located at far distances from each other, their observations can often be considered statistically independent.

\subsection{Computational complexity} \label{subsect: complexity}

As is discussed in \cite{azar2011speedy}, the computational complexity of exact Q-learning is proportional to the size of state-action space. Exact Q-learning is used in the centralized and distributed training phases of SAIC and ESAIC. In the centralized training phase of SAIC, the computational complexity $\mathcal{O}( |\Omega  \times \mathcal{M}|^N)$ grows exponentially with the size of MAS $N$. Accordingly, the addition of each agent to the system multiplies the complexity of the Q-learning by $|\Omega \times \mathcal{M}|$. The complexity $\mathcal{O}( |\Omega  \times \mathcal{M}|^2)$ of the centralized training phase in ESAIC with respect to the size of the MAS $N$, however, is constant time. That is, ESAIC will always execute at the same time (or space) regardless of the size of the MAS $N$. 

The complexity $\mathcal{O}(| \Omega \times \mathcal{C}^{n-1} \times \mathcal{M}|)$ of the Q-learning problem that each agent solves in SAIC, at the decentralized training phase, also grows exponentially with the addition of each agent to the system. Compared with the centralized training phase, in the distributed training phase, SAIC is much less sensitive to the addition of an agent to the system. Although the complexity of the Q-learning at each agent $i$ multiplies by a constant $|\mathcal{C}|$ with the addition of each agent to the system, the size of the communication space $|\mathcal{C}|$ is much smaller than $|\Omega \times \mathcal{M}|$\footnote{To understand why "the size of the communication space $|\mathcal{C}|$ is much smaller than $|\Omega \times \mathcal{M}|$",  remember that we solve the problem (\ref{Task-Based Information Compression - ESAIC}) to significantly reduce the size of the communication message space $\mathcal{C}$ of agent $i$ compared with the size of its observation space $\Omega$.}. In the decentralized training phase, ESAIC follows the same complexity patterns.

\textit{Remark:} If the condition $c_1$ of theorem \ref{theorem: main theorem} is met, ESAIC offers the same performance as SAIC at a much reduced computational cost. Accordingly, for a problem comprised of $N$ agents, the time complexity of SAIC is $| \Omega \times \mathcal{M}|^{N-2}$ times higher than ESAIC.

\section{Numerical Studies} \label{sec: numerical studies}


\textcolor{Mycolor3}{To evaluate our proposed method, ESAIC, in this section, we leverage numerical experiments on a specific cooperative task i.e., a geometric consensus problem with finite observability, called the rendezvous problem. Geometric consensus problems are emerging in many new applications, such as UAV/vehicle platooning, which makes them a useful application domain for the framework proposed in this paper\cite{barel2017come}. Based on the results demonstrated in this section, the proposed framework, ESAIC, has been shown to be a suitable candidate for the distributed control of the large vehicle/UAV platoons under limited communications.}

The rendezvous problem, which is a subcategory of the geometric consensus, has already been studied in the literature of multi-agent systems \cite{zilber2001communication,amato2009incremental}, whereas in our case the inter-agent communication channel is set to have a limited bit-budget. The rendezvous is a particularly interesting testbed for multi-agent communications as it allows us to consider a cooperative MAS consisting of multiple agents whose coordination is dependent on communication. In particular, as detailed in subsection \ref{rendezvous problem - subsection}, if the communication between agents is not efficient, at any time step $t$ each agent $i$ will only have access to its local observation $\ro_i(t)$, which is its own location in the case of rendezvous problem. This mere information is insufficient for an agent to attain the larger reward $C_2$, but is sufficient to attain the smaller reward $C_1$. Accordingly, compared with cases in which no communication between agents is present, in the set-up of the rendezvous problem, efficient communication policies can increase the attained objective function of the MAS \cite{mostaani2022task}. We consider a variety of grid worlds with different size values $N$ and different locations for the goal-point $\omega^T$. We compare the proposed ESAIC and SAIC with the centralized Q-learning scheme that is guaranteed to achieve the optimal average return performance in the rendezvous problem. Our approach can be straightforwardly applied to other geometric systems e.g., by changing the reward function. In particular, a reward function that encourages the agents to come together as close as possible but not collide with each other can emulate a vehicle platooning scenario. While useful, it is outside the scope of our work to investigate the response of the multi-agent system to different rewarding schemes.  Note that, according to \cite{mostaani2022task}, regardless of the definition of the reward function, the geometric consensus problem (or in general the joint quantization and control problem) can be solved by SAIC if the necessary conditions are met, and a centralized training phase is feasible.

\vspace{-0.4cm}
 \subsection{Rendezvous Problem} \label{rendezvous problem - subsection}
 \vspace{-0.4cm}
  \begin{figure}[thb!]
  \centering
      \includegraphics[width=0.55\textwidth]{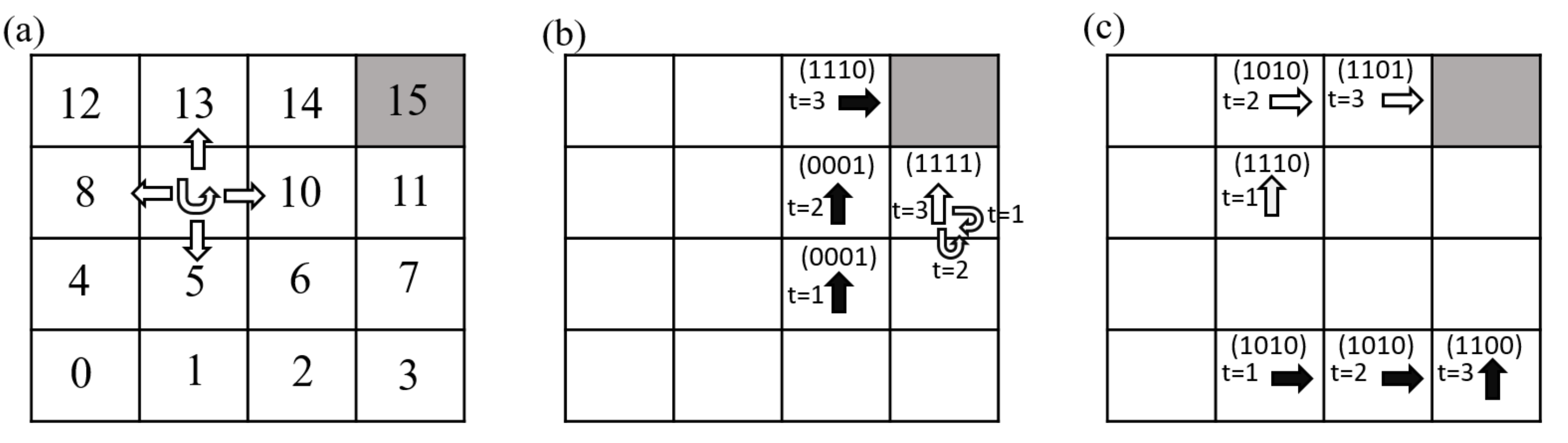}
      \vspace{-0.4cm}
  \caption{The rendezvous problem when $n=2$, $N=4$ and $\omega^T=15$: (a) illustration of the observation space, $ \Omega$, i.e., the location on the grid, and the environment action space $\mathcal{M}$, denoted by arrows, and of the goal state $\omega^T$, marked with gray background; (b) demonstration of a sampled episode, where arrows show the environment actions taken by the agents (empty arrows: actions of agent 1, solid arrows: actions of agent 2) and the $B=4$ bits represent the message sent by each agent. A larger reward $C_2> C_1$ is given to both agents when they enter the goal point at the same time, as in the example; (c) in contrast, $C_1$ is the reward accrued by agents when only one agent enters the goal position \cite{mostaani2019Learning}.}
  \label{fig: rendezvous problem}
  \vspace{-0.4cm}
\end{figure}

 As illustrated in Fig. \ref{fig: rendezvous problem}, in a rendezvous problem, multiple agents operate on an $N \times N$ grid world and aim at arriving at the same time at the goal point on the grid. The system operates in discrete time, with agents taking actions and communicating in each time step $t=1,2,...$ . Each agent $i \in \mathcal{N}$ at any time step $t$ can only observe its own location $\ro_i(t) \in \Omega$ on the grid, where the observation space is $\Omega = \{0,1,...,n^2-1\}$. Each episode terminates as soon as an agent or more visit the goal point which is denoted as $\omega^T \in \Omega$. That is, at any time step $t$ that the observation of each agent $i \in \mathcal{N}$ is a member of $\Omega^T$, the episode will be terminated - so the time horizon $M$ is non-deterministic. The subset $\mathcal{S}^T \subset \mathcal{S}$ also defines all state realizations where one or more agents are in the goal location i.e.,
 
 $\mathcal{S}^T = \{ \langle \ro_1(t), ..., \ro_n(t) \rangle \in \mathcal{S} \, | \,  \exists i \in \mathcal{N} : \ro_i(t) \in \omega^T \}$.
 
 We also define the subset $\mathcal{S}^T_{n'} \subset \mathcal{S}^T$ that includes all the terminal states where only $n'$ number of agents have arrived at the goal location i.e.,
 
  $\mathcal{S}^T_{n'} = \{ \langle \ro_1(t), ..., \ro_n(t) \rangle \in \mathcal{S} \, | \,  \forall i \in \mathcal{N}' : \ro_i(t) \in \omega^T \}$,
  
  where $\mathcal{N}' \subseteq \mathcal{N}$ is a subset of all agents with size $| \mathcal{N}' | = n'$. Following the same definition for $\mathcal{S}^T_{n'}$, the subset $\mathcal{S}^T_n$ is equivalent to the set of all terminal states where all agents are at the goal location. At time $t=1$, the initial position of all agents is randomly and uniformly selected amongst the non-goal states, i.e., for each agent $i \in \mathcal{N}$ the initial position of the agent is $\ro_i(1) \in \Omega - \{ \omega^T \}$.

At any time step $t=1,2,...$ each agent $i$ observes its position, or environment state, and acquires information about the position of the other agents by receiving a communication message vector $ {\rc}_{-i}(t)$ sent by the other agents $j \in \mathcal{N}_{-i}$ at the time step $t$.  Based on this information, agent $i$ selects its environment action $\rm_i(t)$ from the set $\mathcal{M} = \{\text{Right},\text{Left},\text{Up},\text{Down},\text{Stop}\}$, where an action $\rm_i(t) \in \mathcal{M}$ represent the horizontal/vertical move of agent $i$ on the grid at time step $t$. For instance, if an agent $i$ is on a grid-world as depicted on Fig. \ref{fig: rendezvous problem} (a), and observes $\ro_i(t)=4$ and selects "Up" as its action, the agent's observation at the next time step will be $\ro_i(t+1)=8$. If the position to which the agent should be moved is outside the grid, the environment is assumed to keep the agent in its current position. We assume that all these deterministic state transitions are captured by $T\big(\ro_1(t), ..., \ro_n(t),\rm_1(t), ..., \rm_n(t)\big)$, which can determine the observations of agents in the next time step $t+1$ following
{\small\[
\langle \ro_1(t+1), ..., \ro_n(t+1) \rangle
= T\big(\ro_1(t), ..., \ro_n(t),\rm_1(t), ..., \rm_n(t)\big).
\]}
Accordingly, given observations $\langle \ro_i(t+1), ..., \ro_n(t+1) \rangle$ and actions $ \langle \rm_1(t+1), ..., \rm_n(t+1) \rangle $, all agents receive a single team reward
{\small\begin{equation}\label{example environment noise distribution}
   r\big( \ro_1(t), ..., \ro_n(t),\rm_1(t), ..., \rm_n(t) \big)=
   \begin{cases}
    C_1, & \text{if  $P_1$}\\
    C_2, & \text{if  $P_2$},\\
    0, & \text{otherwise},\\
   \end{cases}
\end{equation}}
\sloppy where $C_1 < C_2$ and the propositions $P_1$ and $P_2$ are defined as $P_1: T\big(\ro_1(t), ..., \ro_n(t),\rm_1(t), ..., \rm_n(t)\big) \in \mathcal{S}^T - \mathcal{S}^T_n$ and $P_2: T\big(\ro_1(t), ..., \ro_n(t),\rm_1(t), ..., \rm_n(t)\big) \in \mathcal{S}^T_n$. When only a subset $\mathcal{N}',\,\, |\mathcal{N}| = n' < n$ of agent arrives at the target point $\omega^T$, the episode will be terminated with the smaller reward $C_1$ being obtained, while the larger reward $C_2$ is attained only when all agents visit the goal point at the same time. Note that this reward signal encourages coordination between agents which in turn can benefit from inter-agent communications. 

Furthermore, at each time step $t$ agents choose a communication vector to communicate with every other agent $j \in \mathcal{N}_{-i}$ by selecting a $\rc_{i,j}(t) \in \mathcal{C} = \{0,1\}^{R_{i,j}}$ of $R_{i,j}$ bits, where $R_{i,j}$ (bits per channel use / per time step) is the fixed bit-budget of the inter-agent communication channel between agent $i$ and $j$.
 The goal of the MAS is to maximize the average return by solving the D-JCCD problem (\ref{eq: Decentralized Joint Control and Communication Design (D-JCCD) problem}).

\subsection{Results}

ESAIC, SAIC, and centralized schemes are compared by their average return in Fig. \ref{fig: average return vs training iterations - SAIC ESAIC Centralized}. The figure is intended to show the applicability of the ESAIC scheme in more complex geometric consensus environments. The size of the grid world for this figure is $8 \times 8$, and the multi-agent system is composed of three agents. The figure demonstrates that the performance of ESAIC closely follows that of SAIC, with almost similar average return performance as well as the speed of convergence. The centralized scheme, which is represented by the solid black curve, achieves optimal performance but requires virtually twice the time required for the convergence of ESAIC and SAIC. Fig. \ref{fig: average return vs training iterations - SAIC ESAIC Centralized} suggests that ESAIC is a promising approach for achieving high average return performance in complex MASs, with similar performance to SAIC and faster convergence time than the centralized scheme.

    \begin{figure}[ht] 
      \centering 
          \includegraphics[width=0.6\textwidth]{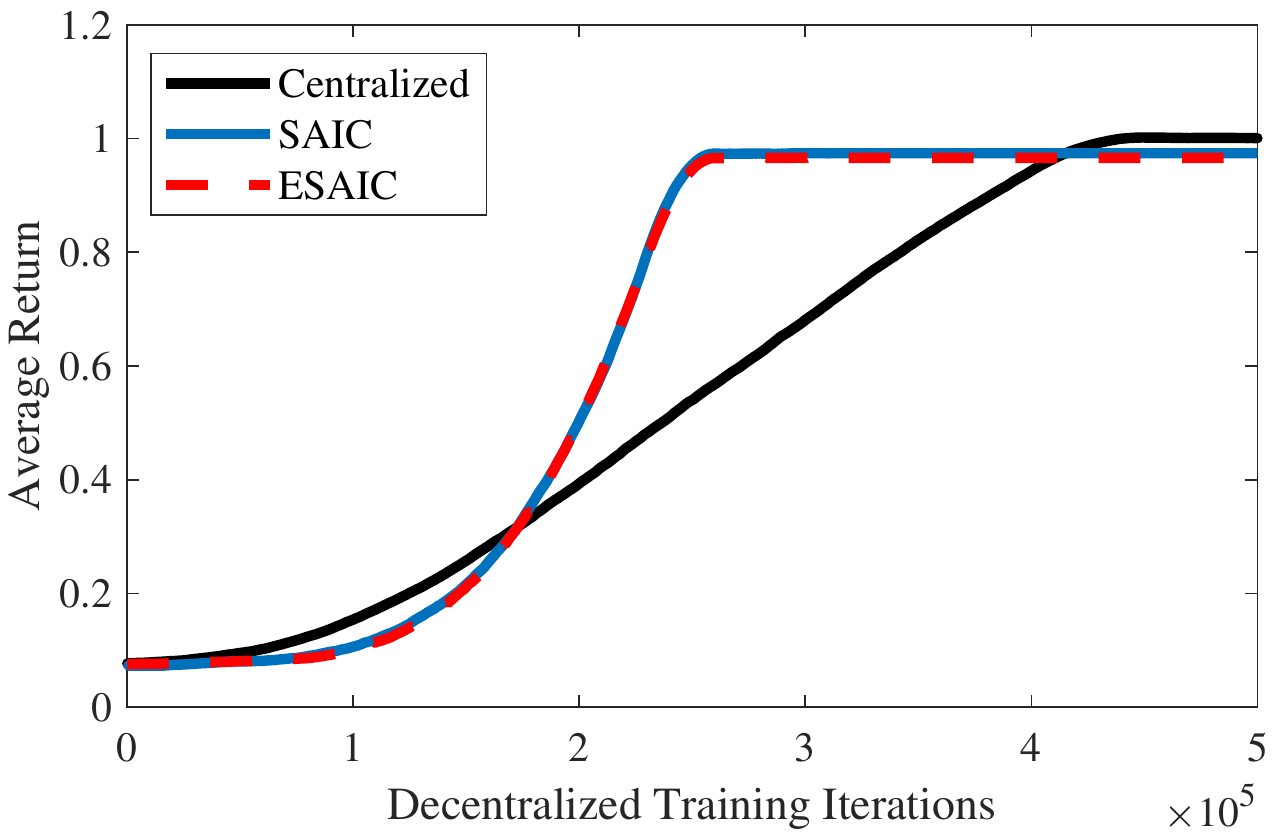} 
          \vspace{-4mm}
      \caption{Comparison of the obtained average return via SAIC and ESAIC in MAS in the decentralized training phase while the condition $c_1$ in (\ref{condition: main condition--}) is violated.}
      \label{fig: average return vs training iterations - SAIC ESAIC Centralized}
      \vspace{-0.4cm}
    \end{figure}

    Figure \ref{fig: average return vs training iterations - SAIC ESAIC Centralized} was comparing the average return performance of ESAIC against SAIC for a three-agent system. The following figure, Figure \ref{fig: return_vs_agents}, presents a similar comparison for multi-agent systems with a variable number of agents. The results shows that ESAIC achieves an average return performance that is similar to SAIC, while also offering a remarkable reduction in computational complexity. Due to its extravagant computational complexity, SAIC could not be evaluated for multi-agent systems composed of more than 4 agents. Given the exponential increase in the complexity of SAIC with respect to the number of agents, to be able to study ints performance for a 4-agent system, this figure has been plotted for the grid worlds of smaller size i.e., $3 \times 3$ across all schemes.

    \begin{figure}[ht] 
      \centering 
          \includegraphics[width=0.6\textwidth]{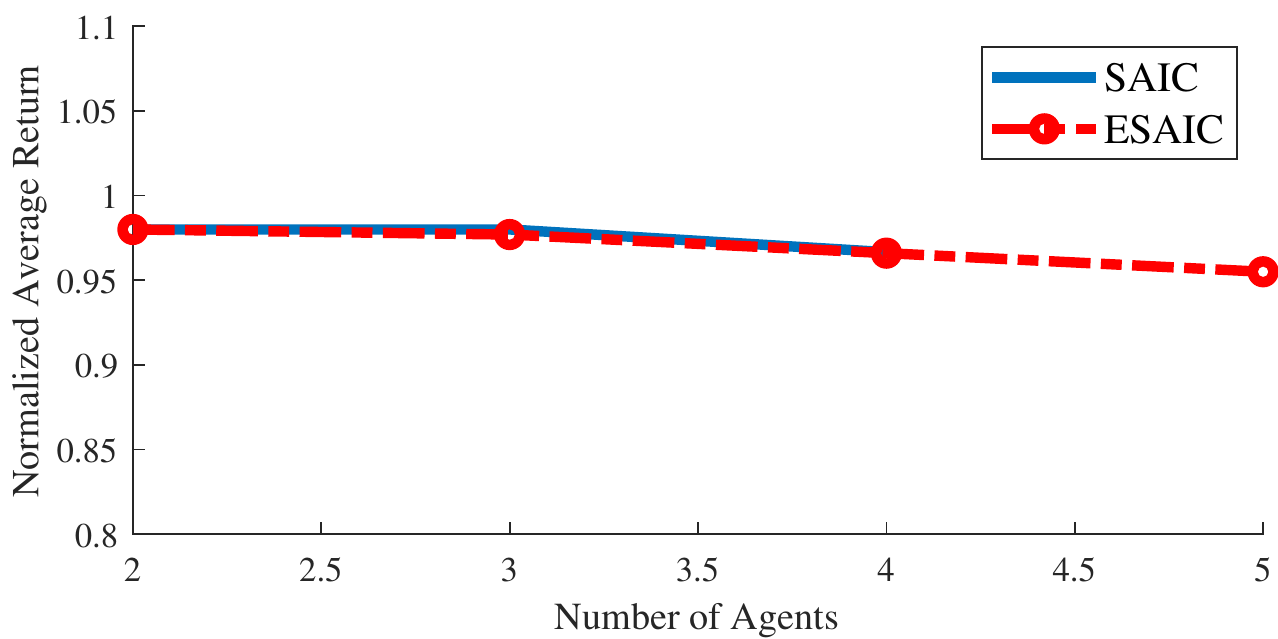} 
          \vspace{-4mm}
      \caption{Comparison of the obtained average return via SAIC and ESAIC in MAS with varying numbers of agents.}
      \label{fig: return_vs_agents}
      \vspace{-0.4cm}
    \end{figure}

    As discussed earlier in section \ref{sec: analytical studies}, SAIC suffers from prohibitively high computational complexity in its centralized training phase. ESAIC is introduced in this paper to tackle the issue of complexity in the centralized training phase by designing the communication policies only according to a two-agent centralized training. Figure \ref{fig: centralized_time} compares the run time required for the implementation of the centralized training phase in both schemes SAIC and ESAIC - both theoretically and analytically. Similar to Fig. \ref{fig: return_vs_agents}, this figure as well as the next one have been plotted for the grid worlds of smaller size i.e., $3 \times 3$ across all schemes. The analytical results reflect the explanations provided at \ref{subsect: complexity}.

    \begin{figure}[h] 
      \centering 
          \includegraphics[width=0.6\textwidth]{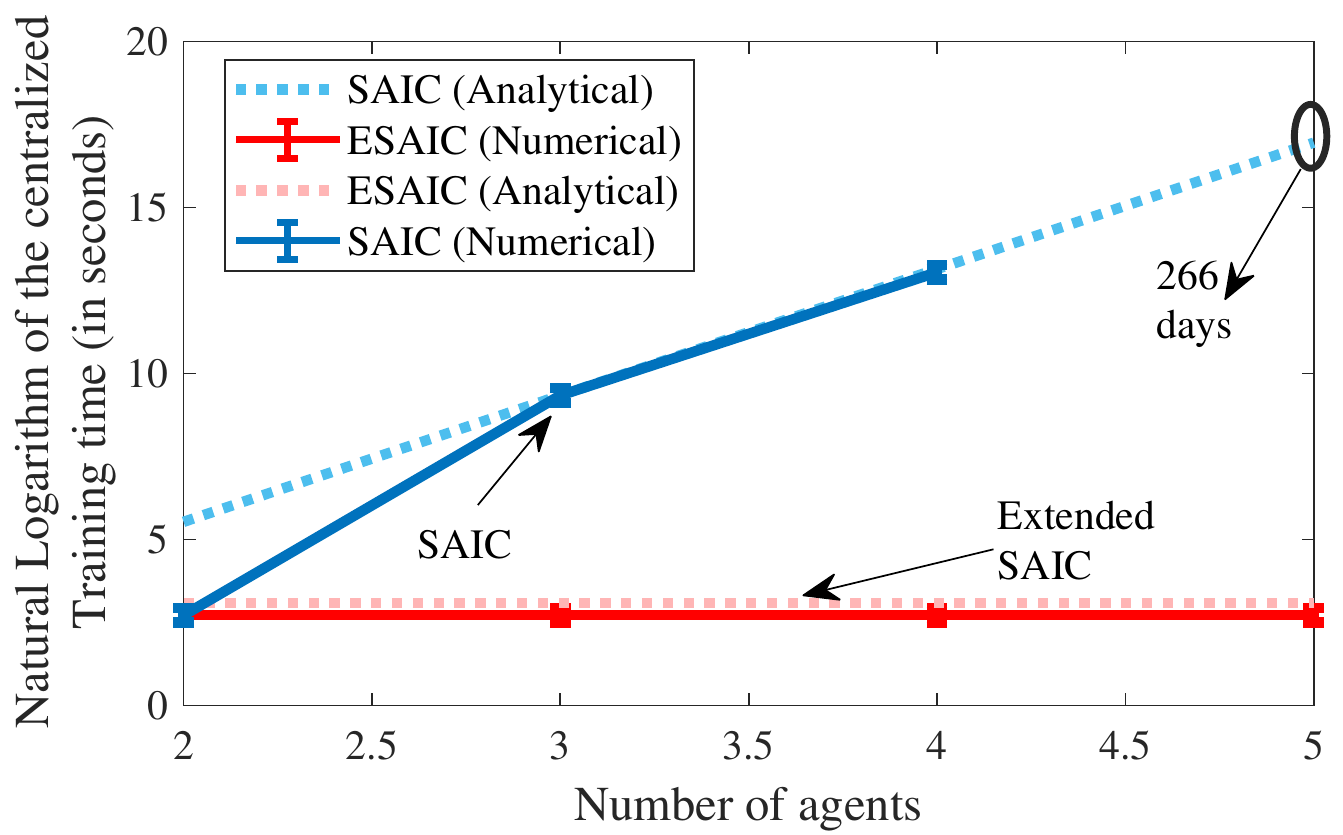} 
          \vspace{-3mm}
      \caption{Comparison of the average time required to carry out the centralized training phase in both algorithms SAIC and ESAIC.}
      \label{fig: centralized_time}
      \vspace{-0.4cm}
    \end{figure}

    To realize the end-to-end time required for the training of both algorithms, Fig. \ref{fig: end to end_time} is brought. This figure illustrates the combined time required to carry out the centralized as well as the decentralized training phase. Inter-agent communications are considered to be $R_{i,j} = 2$ (bits per channel use) across all agents $\forall i, j \in \mathcal{N}$. With an increase in the number of agents, the size of the received communication message space $\mathcal{C}^{n-1}$ increases exponentially leading to an increase in the end-to-end complexity of both algorithms SAIC and ESAIC. Nevertheless, the goal of solving the problems (\ref{Task-Based Information Compression}) and (\ref{Task-Based Information Compression - ESAIC}) is to significantly reduce the size of each agent's communication transmission space $\mathcal{C}$ compared with the observation space $\Omega$. Accordingly, the exponential increase in the size of received communication space has a much less pronounced impact on the overall complexity of both algorithms. Yet, we expect the size of the received message space $\mathcal{C}^{n-1}$ to be another bottleneck of SAIC that ESAIC can not solve. This bottleneck gets more serious when the number of agents goes double digits. The analytical results reflect the explanations provided at \ref{subsect: complexity}.

    \begin{figure}[h] 
      \centering 
          \includegraphics[width=0.6\textwidth]{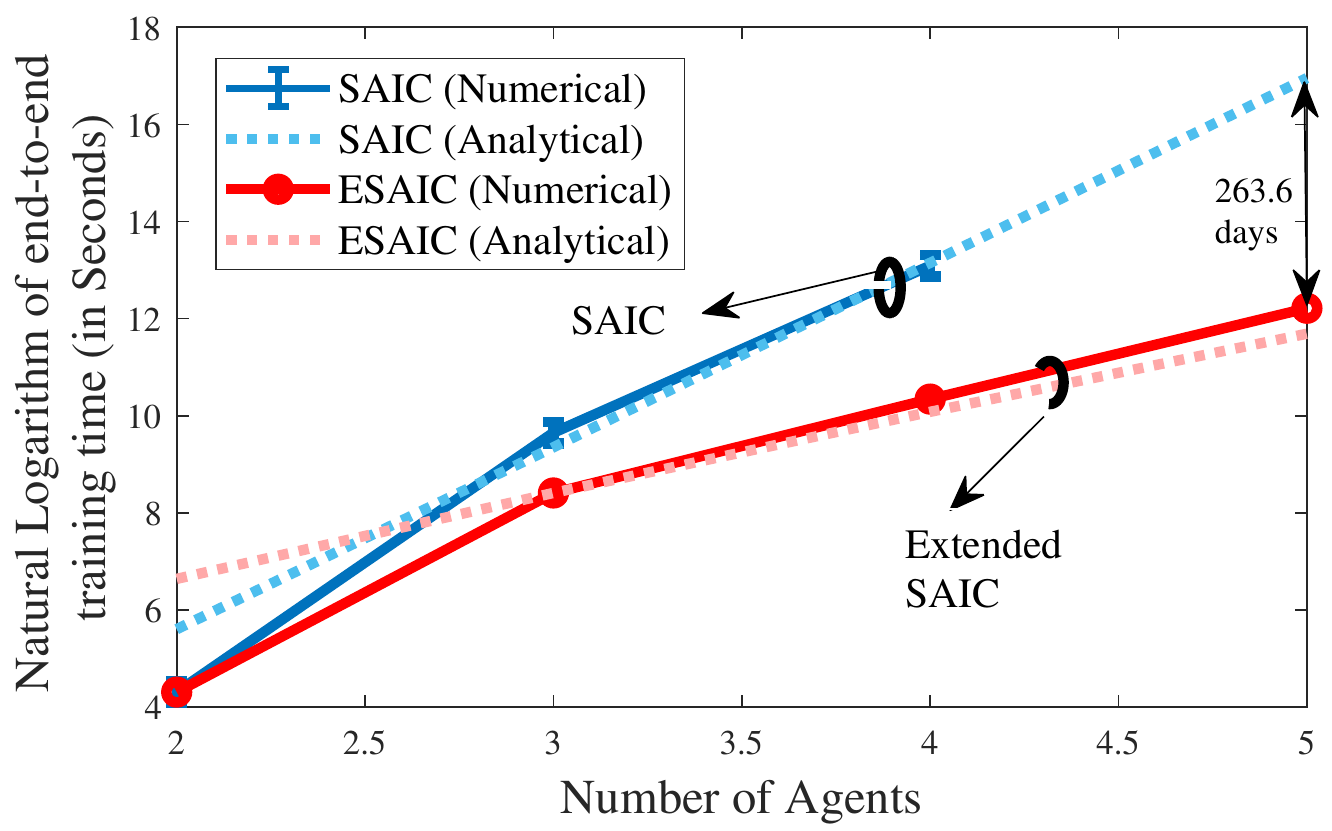} 
          \vspace{-4mm}
      \caption{Comparison of the average time required to carry out end-to-end training in both algorithms SAIC and ESAIC.}
      \label{fig: end to end_time}
      \vspace{-0.3cm}
    \end{figure}

    To show that both SAIC and ESAIC can perform well even under heterogeneous bit-budgets, Fig. \ref{fig: heterogeneous bit-budgets} is obtained. This figure studies the average return performance of ESAIC - which is equivalent to that of SAIC for a two-agent system - in an $8 \times 8 $ grid world with heterogeneous bit-budgets for agents. We observe near-optimal performance for both schemes for all heterogeneous rates $R_{i,j} > 2$.

    \begin{figure}[h] 
      \centering 
          \includegraphics[width=0.6\textwidth]{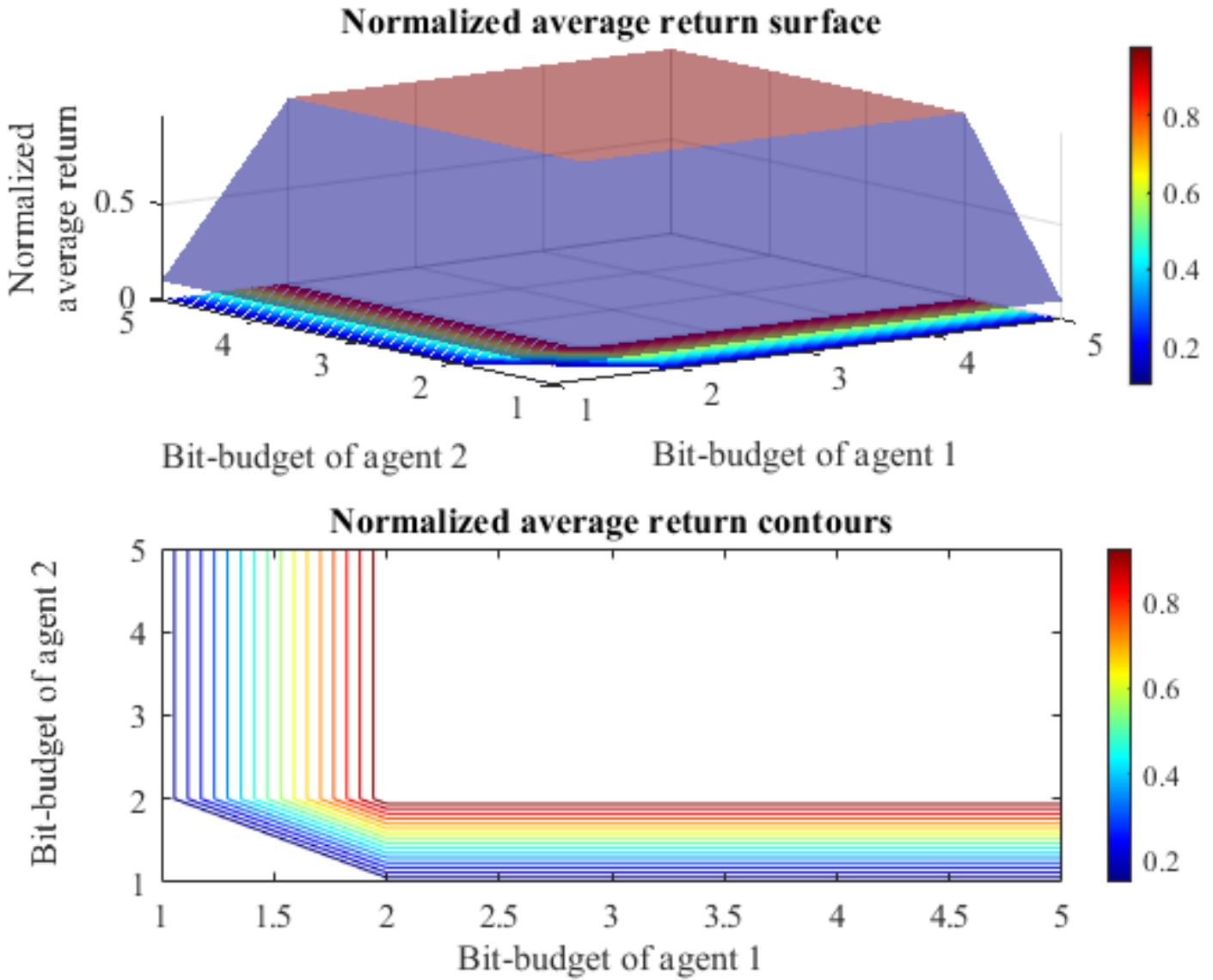} 
          \vspace{-4mm}
      \caption{Normalized average return of a two-agent system when ESAIC is applied under heterogeneous bit-budgets.}
      \label{fig: heterogeneous bit-budgets}
      \vspace{-0.4cm}
    \end{figure}

    \vspace{-0.4cm}
    \section{Conclusion}
    \vspace{-0.4cm}
    This paper presented a novel scalable task-oriented quantization algorithm for multi-agent communications over bit-budgeted channels. The proposed algorithm, ESAIC, offers a unique approach to designing the communications of a multi-agent system, regardless of the number of agents involved. The two-agent centralized training phase used in the algorithm has been shown to be effective in designing abstract communications and obtaining near-optimal average returns. We have also demonstrated that any approximate/deep reinforcement learning scheme can be used in the centralized training phase without affecting the reported results - making the results of the paper more applicable to real-world scenarios. The results of our analytical analysis are strong evidence for the effectiveness of the proposed algorithm. The main theorem proposed by the paper offers a solid foundation for future research as we expect the condition of the theorem to be further relaxed in future works, leading to even more efficient and effective communication designs for large multi-agent systems. We believe that the proposed algorithm, ESAIC, has significant implications for the design of communication systems in multi-agent systems, with potential applications in areas such as autonomous vehicles, robotics, and wireless sensor networks.

\bibliographystyle{IEEEtran}
\vspace{-8mm}
{\small
\bibliography{Bibfile}}
\vspace{-4mm}
\appendices





\vspace{-2mm}
\section{Proof of Theorem \ref{theorem: main theorem}} \label{app: proof: theorem: theorem: main theorem}
\vspace{-4mm}
To prove this theorem we first introduce a lemma together with its proof in subsection \ref{lemma: instrumental}. We then use the result of this lemma in subsection \ref{app:subsect: lemma: oracle} to see how one can design the partition $\mathcal{P}^{[N]}_{i,j}$, after a two-agent centralized training but given the help of an oracle that knows a specific function $f(\cdot)$. Subsequently, we complete the proof of Theorem \ref{theorem: main theorem}, in subsection \ref{app: subsect: proof: theorem: main theorem} leveraging the above-mentioned lemmas with no further need to the knowledge of the function $f(\cdot)$.
\vspace{-3mm}
\subsection{An Instrumental Lemma} \label{lemma: instrumental}
\vspace{-3mm}
\begin{lemma} \label{lemma: function acting over a set}
Let $\mathcal{A} \subset \mathbb{R}$ and $\mathcal{B}\subset \mathbb{R}$ be discrete sets, $h(\cdot) : \mathbb{R} \rightarrow \mathbb{R}$ be a bijection, $\Ddot{h}(\mathcal{A}) $
be a discrete set and $c \in \mathbb{R}$ be a constant value. We can state
\vspace{-5mm}
\small{ \begin{align}
   \Ddot{h}(\mathcal{A}) = \mathcal{B} \implies \sum_{a' \in h(\mathcal{A})} |a' - c| = \sum_{b \in \mathcal{B}} |b - c|.
\end{align} }
\begin{proof}
\vspace{-8mm}
    \small{ \begin{align}
     &    \forall \, b \in \mathcal{B} \,\, \exists a\in \mathcal{A} : h (a) = b,
    \end{align} }
    where adding the constant value $-c$ to the side of equality and applying the absolute value will result in
    \vspace{-7mm}
    \small{ \begin{align}
        \forall \, b \in \mathcal{B} \,\, \exists a\in \mathcal{A} : |h (a) - c| = |b - c|.
    \end{align} }
    Which is equivalent to 
    \vspace{-7mm}
    \small{ \begin{align}
        \forall \, b \in \mathcal{B} \,\, \exists a'\in \ddot{h}(\mathcal{A}) : a' - c = b - c,
    \end{align} }
    according to the definition of $\ddot{ h}(\mathcal{A})   \triangleq \{ a' : h^{-1}(a') \in \mathcal{A} \} $.
    Performing a summation across all the elements of $\ddot{ h}(\mathcal{A})$ and $\mathcal{B}$ will result in
    \vspace{-4mm}
    \small{ \begin{align}
        \sum_{a' \in \ddot{ h}(\mathcal{A})} |a' - c| = \sum_{b \in \mathcal{B}} |b - c|.
    \end{align} }
\end{proof}
\end{lemma}
\vspace{-10mm}
\subsection{If an oracle tells us $f(\cdot)$} \label{app:subsect: lemma: oracle}

\begin{lemma} \label{lemma: equivalence of SAIC and ESAIC - in rate constant comms}
 Let the bijection $f(\cdot): \mathcal{V}^{[2]} \rightarrow \mathcal{V}^{[N]} $ be the mapping from the value of observations for a two-agent scenario to the $N$-agent. For all $i,j \in \mathcal{N}$, the partition $\mathcal{P}^{[2]}_{i,j}$ proposed by ESAIC are the same as the partition $\mathcal{P}^{[N]}_{i,j}$ proposed by SAIC if the function $f(\cdot)$ is known and if
 \vspace{-4mm}
 \small{ \begin{align}
     &  c_1: \label{condition: main condition}
     & & \forall k \in \{1,\dots, B_{i,j}\} \,\,\, \exists k' \in \{1,\dots, B_{i,j}\}  : \\
     & & & \Ddot{f}(\mathcal{P}'^{[2]}_{i,j,k}) = \mathcal{P}'^{[N]}_{i,j,k'} \notag 
 \end{align} }
\end{lemma}
\begin{proof}
    We start by a clustering problem over the space of values $\mathcal{V}^{[2]}$ that is obtained by ESAIC and we show the problem as well as its solution are equivalent to the problem that is solved by SAIC.
    
    Given the help of an oracle, we know the function $f(\cdot)$ and thus, after obtaining the values $ V^{* \, [2]}(\ro) =  v^{[2]} \in \mathcal{V}^{[2]}$ by solving the centralized two agent problem, we proceed by obtaining the solution for
    \vspace{-7mm}
    \small{ \begin{align}\label{eq: clustering in ESAIC - v2}
        & \underset{\mathcal{P}'}{\text{argmin }}
        & & \sum_{k' = 1}^{2^R} \,\,\, \sum_{v^{[2]} \in \mathcal{P}^{'  \,\,\, [2]}_{i,j,k'}} | f(v^{[2]}) - \mu_{k'} |.
    \end{align} }
    We also know from (\ref{condition: main condition}) and lemma \ref{lemma: function acting over a set} that for any $k' \in \{1,...,2^R\}$ there is a $k \in \{1,...,2^R\}$ such that
    \vspace{-7mm}
    \small{ \begin{align}\label{eq: equality}
        \sum_{v^{[2]} \in \mathcal{P}^{'  \,\,\, [2]}_{i,j,k'}} | f(v^{[2]}) - \mu_{k'} | = 
        \sum_{v^{[N]} \in \mathcal{P}^{'  \,\,\, [N]}_{i,j,k}} | v^{[N]} - \mu_{k} |. 
    \end{align} }
    By replacing the right-hand term in equality (\ref{eq: equality}) with the inner summation of eq. (\ref{eq: clustering in ESAIC - v2}), we will arrive at
    \vspace{-7mm}
    \small{ \begin{align} \label{eq: SAIC clustering--}
        & \underset{\mathcal{P}}{\text{argmin }}
        & & \sum_{k = 1}^{2^R} \,\,\, \sum_{ v^{[N]} \in \mathcal{P}'^{[N]}_{i,j,k}} | v^{ [N]} - \mu_{k} |,
    \end{align} }
    and since problem (\ref{eq: SAIC clustering--}) is the exact problem that is solved by SAIC, the proof is concluded.

\end{proof}

\vspace{-10mm}
\subsection{Proof of theorem \ref{theorem: main theorem}} \label{app: subsect: proof: theorem: main theorem}
\begin{proof}

    As we assume that we have no knowledge about the function $f(\cdot)$, after obtaining the optimal value function $V^{* [2]}(\cdot)$, we will solve the clustering problem as if we are designing a communication policy for a two-agent system by SAIC. Accordingly, we will have to solve
    \vspace{-4mm}
    {\small \begin{align}\label{eq: clustering in ESAIC - v1}
        & \underset{\mathcal{P}'}{\text{argmin }}
        & & \sum_{k' = 1}^{2^R} \,\,\, \sum_{v^{[2]} \in \mathcal{P}^{'  \,\,\, [2]}_{i,j,k'}} | v^{[2]} - \mu_{k'} |.
    \end{align}}
     Given eq. (\ref{eq: equality}) and lemma \ref{lemma: function acting over a set}, we know that for any $k' \in \{1,...,2^R\}$ there is a $k \in \{1,...,2^R\}$ such that
     \vspace{-7mm}
    {\small \begin{align}\label{eq: equality - v2}
        \sum_{v^{[2]} \in \mathcal{P}^{'  \,\,\, [2]}_{i,j,k'}} | f^{-1}\big( f(v^{[2]}) \big) - \mu_{k'} | = 
        \sum_{v^{[N]} \in \mathcal{P}^{'  \,\,\, [N]}_{i,j,k}} | f^{-1} (v^{[N]}) - \mu_{k} |. 
    \end{align} }
        Be reminded that the inner summation of eq. (\ref{eq: clustering in ESAIC - v1}) is equal to the left-hand term in equality (\ref{eq: equality - v2}). By replacing the right-hand term in equality (\ref{eq: equality - v2}) with the inner summation of eq. (\ref{eq: clustering in ESAIC - v1}), we will arrive at
    \vspace{-4mm}
    {\small \begin{align} \label{eq: SAIC clustering}
        & \underset{\mathcal{P}}{\text{argmin }}
        & & \sum_{k = 1}^{2^R} \,\,\, \sum_{ v^{[N]} \in \mathcal{P}^{[N]}_{i,j,k}} | f^{-1} \big( v^{ [N]}\big) - \mu_{k}  |.
    \end{align} }
    The inner summation of eq. (\ref{eq: SAIC clustering}) can also be taken over the observation space as the following
    \vspace{-4mm}
    {\small \begin{align} \label{eq: SAIC clustering--}
        & \underset{\mathcal{P}}{\text{argmin }}
        & & \sum_{k = 1}^{2^R} \,\,\, \sum_{ \ro \in \mathcal{P}^{[N]}_{i,j,k}} | f^{-1} \big( V^{* [N]}(\ro)\big) - \mu_{k}  |,
    \end{align} }
    where by applying the function $f(\cdot)$, according to the lemma \ref{lemma: function acting over a set}, we will get
    \vspace{-4mm}
    {\small \begin{align} \label{eq: SAIC clustering---}
        & \underset{\mathcal{P}}{\text{argmin }}
        & & \sum_{k' = 1}^{2^R} \,\,\, \sum_{ \ro \in \mathcal{P}^{[2]}_{i,j,k'}} | V^{* [2]}(\ro) - \mu_{k'}  |.
    \end{align} }

    
    From this point onward using a set of known relationships we will try to find the relationship between $\mathcal{P}^{[2]}_{i,j,k'} $ and $ \mathcal{P}^{[N]}_{i,j,k}$. It will be demonstrated that $\mathcal{P}^{[2]}_{i,j,k'} = \mathcal{P}^{[N]}_{i,j,k}$ where $k$ is the same index that allows the equality $\ddot{ f}(\mathcal{P'}^{[2]}_{i,j,k'}) =  \mathcal{P'}^{[N]}_{i,j,k}$ to hold. Then by running the inner summation of eq. \ref{eq: SAIC clustering} over all the elements of the set $ \mathcal{P'}^{[N]}_{i,j,k}$ - instead of $ \mathcal{P'}^{[2]}_{i,j,k'}$ - the proof will be concluded.
    
    We know that the the set $\mathcal{P}^{[2]}_{i,j,k'}$ corresponds to $\mathcal{P}'^{[2]}_{i,j,k'}$, i.e.,
    \vspace{-4mm}
    {\small \begin{align} \label{eq: from observation clusters v[2] to value v[2] clusters}
        \ddot{ V}^{*\,[2]}\big( \mathcal{P}^{[2]}_{i,j,k'} \big) = \mathcal{P'}^{[2]}_{i,j,k'}.
    \end{align} }
    Despite being a surjection, we define the \textit{inverse image} of $V^{*\,[2]}(\cdot)$ to be 
    \vspace{-4mm}
    {\small \begin{align} \label{eq: mapping observations to v[2] values}
           \big[ V^{*\,[2]} \big]^{-1}(v^{[2]}) \triangleq \{ \, \ro \, | \,  V^{*\,[2]}(\ro) = v^{[2]} \},
    \end{align} }
    such that
    \vspace{-8mm}
    {\small \begin{align} \label{eq: mapping observation space to v[2] value space}
        \big[\ddot{ V}^{*\,[2]} \big]^{-1}\big( \mathcal{P}'^{[2]}_{i,j,k'}\big) = \mathcal{P}^{[2]}_{i,j,k'} .
    \end{align} }
     Similar to (\ref{eq: mapping observations to v[2] values}) we define the inverse of $V^{*\,[N]}(\cdot)$ to be
     \vspace{-4mm}
    {\small \begin{align} \label{eq: mapping observations to v[N] values}
           \big[ V^{*\,[N]} \big]^{-1}(v^{[N]}) = \{ \, \ro \, | \,  V^{*\,[N]}(\ro) = v^{[N]} \},
    \end{align}}
    such that
    \vspace{-4mm}
    {\small \begin{align} \label{eq: mapping observation space to v[N] value space}
        \big[\ddot{ V}^{*\,[N]} \big]^{-1}\big( \mathcal{P}'^{[N]}_{i,j,k}\big) = \mathcal{P}^{[N]}_{i,j,k} .
    \end{align} }
         To show the equivalnce of $\mathcal{P}^{[2]}_{i,j,k'} = \mathcal{P}^{[N]}_{i,j,k}$, we will show that they both correspond to the same value cluster.  
        Further to the (\ref{condition: main condition}), we know that there exit a $k$ such that $\ddot{ f}(\mathcal{P'}^{[2]}_{i,j,k'}) =  \mathcal{P'}^{[N]}_{i,j,k}$. This together with the eq. (\ref{eq: mapping observation space to v[N] value space}) and (\ref{eq: from observation clusters v[2] to value v[2] clusters}) implies that
        \vspace{-4mm}
        {\small \begin{align}
           \big[\ddot{ V}^{*\,[N]} \big]^{-1} \Big( \ddot{f} \big( \ddot{ V}^{*\,[2]}( \mathcal{P}^{[2]}_{i,j,k'} ) \big)  \Big) = \mathcal{P}^{[N]_{i,j,k}}.
        \end{align}}
        In the following we will show that $ \big[\ddot{ V}^{*\,[N]} \big]^{-1} \Big( \ddot{f} \big( \ddot{ V}^{*\,[2]}( \cdot ) \big)  \Big)$ is an identity image function; i.e.,
        \vspace{-4mm}
        {\small \begin{align}\label{eq: fog-1 for function images}
            \big[\ddot{ V}^{*\,[N]} \big]^{-1} \Big( \ddot{f} \big( \ddot{ V}^{*\,[2]}( \Omega' \subset \Omega ) \big)  \Big) = \Omega'.
        \end{align}}
        Since we know from the assumptions of lemma \ref{lemma: equivalence of SAIC and ESAIC - in rate constant comms} that $ V^{*\,[N]} (\ro) = f\big(V^{*\,[2]} (\ro)\big) \;\, \forall o \in \Omega$, it could be trivial to show the correctness of $V^{{*\,[N]}^{-1}} \Big( f\big(V^{*\,[2]} (\ro)\big) \Big) = \ro \;\, \forall o \in \Omega$ or (\ref{eq: fog-1 for function images}); but it is not the case where the functions involved are not bijections and when instead of functions we have image functions. Yet, to prove (\ref{eq: fog-1 for function images}), it is sufficient to prove that for all $g(\cdot): \mathcal{E} \rightarrow \mathcal{F}$ and $\mathcal{E}' \subset \mathcal{E}$ we can state that
        \vspace{-6mm}
        {\small\begin{align}
            \ddot{g}^{-1} \Big( \ddot{g} \big( \mathcal{E}' \big) \Big) = \mathcal{E}'.
        \end{align}}
        In the following, we prove the above-mentioned - while a similar statement is can be found in \cite{topology2023Munkres}(page 17) without proof given that $[V^{{*\,[N]}}]^{-1}(\cdot)$ is a surjective mapping.
        According to the definition of image functions, for every image function $g(\cdot): \mathbb{P}(\mathcal{E}) \rightarrow \mathbb{P}(\mathcal{F})$ and $\mathcal{E}' \subset \mathcal{E}$ we have
        \vspace{-6mm}
        {\small \begin{align} \label{eq: inverse definition}
            \epsilon'' \in \ddot{g}^{-1}(\mathcal{E}'') \iff \ddot{g}(\epsilon'') \in \mathcal{E}''.
        \end{align}}
        To show the equality of the set $\ddot{g}^{-1} \Big( \ddot{g} \big( \mathcal{E}' \big) \Big)$ and $\mathcal{E}'$ we now have to prove that
        \vspace{-4mm}
        {\small\begin{align} \label{eq: inverse proof 1}
            & \epsilon' \in \mathcal{E}' 
            & & \iff \epsilon' \in \ddot{g}^{-1} \Big( \ddot{g} \big( \mathcal{E}' \big) \Big)
        \end{align}}
        Consider eq. (\ref{eq: inverse definition}) and replace the $\ddot{g} \big( \mathcal{E}' \big)$ in eq. (\ref{eq: inverse proof 1}) with $\mathcal{E}''$ of eq. (\ref{eq: inverse definition})
        \vspace{-4mm}
        {\small\begin{align}
             & & & \iff \ddot{g} (\epsilon') \in  \ddot{g} \big( \mathcal{E}' \big) 
             & & & \iff \epsilon' \in \mathcal{E}'.
        \end{align}}
\end{proof}

\vspace{-12mm}
\section{Proof of Lemma \ref{lemma: verfiablity}} \label{appendix: proof: lemma: verifiability}
\vspace{-4mm}
Without loss of generality, instead of proving Lemma 3 for strictly monotonic functions, we only prove it for the strictly increasing functions. The strictly decreasing functions also can be treated in a similar manner. Then, Lemma 3 for strictly monotonic functions $f^{[k]}(\cdot)$ is automatically deduced. So, we consider the case $f^{[3]}(\cdot)$ is strictly increasing and we prove $f^{[N]}(\cdot)$, $N > 3$ is also strictly increasing. This is done via induction.
\begin{proof}
    \textbf{Base case:} We assume that $P(3)$ has been verified.

    \textbf{Induction step:} We will show that for every $k > 3$, if $P(k)$ holds, then $P(k + 1)$ also holds. Equivalently, 
\vspace{-4mm}
{\small \begin{align} \label{eq: induction: hypothesis-1}
& P(k):
&&        \text{ If } V^{[2]}(\ro') > V^{[2]}(\ro) \implies \\
& && f^{[k]}\big(V^{[2]}(\ro')\big) > f^{[k]}\big(V^{[2]}(\ro)\big), ~ \forall \ro \in \Omega,
\end{align}}
which can also be expressed as
\vspace{-4mm}
{\small\begin{align} \label{eq: induction: hypothesis}
& P(k):
&&        \text{ If } V^{[2]}(\ro') > V^{[2]}(\ro) \implies V^{[k]}(\ro') > V^{[k]}(\ro), ~ \forall \ro \in \Omega.
\end{align}}

    To prove the statement $P(k+1)$, we represent $V^{[k+1]}(\cdot)$ in terms of $V^{[k]}(\cdot)$ via eq. \eqref{eq: vk} - (\ref{eq: relationship of vk+1 and vk}), and we will use this relationship to deduce the induction step given (\ref{eq: induction: hypothesis}) holding. 
  
The value of observation $\ro \in \Omega$ in an $k$ agent system is expressed by
\vspace{-4mm}
{\small\begin{equation} \label{eq: vk}
    V^{[k]}(\ro_i = \ro) = \sum_{\ro_{-i}\in \Omega^{k-1}} p(\ro_{-i} = [\ro_j]_{j \in \mathcal{N}_{-i}} ) \mathbb{E} \{ \bg^{[k]} | \ro_i , \ro_{-i}\},
\end{equation}}
\sloppy where the set of agents $\mathcal{N} = \{1,2,...,k\}$ is comprised of $k$ elements and $\bg^{[n]}(t^{'})= {\sum}_{t=t^{'}}^{T'}\gamma^{t-1} r^{[n]}\big(\bs(t),\bm(t)\big)$.
Similarly, the value of observation $\ro \in \Omega$ in an $k+1$ agent system is expressed by
\vspace{-4mm}
{\small\begin{equation}
    V^{[k+1]}(\ro_i = \ro) = \sum_{\ro_{-i}\in \Omega^{k}} p(\ro_{-i} = [\ro_j]_{j \in \mathcal{N'}_{-i}} ) \mathbb{E} \{ \bg^{[k+1]} | \ro_i , \ro_{-i}\},
\end{equation}}
where the set of agents $\mathcal{N'} = \{1,2,...,k+1\}$ is comprised of $k+1$ elements. Taking condition $c_2$ into account, it can be readily shown that
\vspace{-4mm}
{\small\begin{equation}\label{eq: value linearized}
    V^{[k+1]}(\ro_i = \ro) = \!\! \sum_{\ro_{-i}\in \Omega^{k}}  \! p(\ro_{-i} = [\ro_j]_{j \in \mathcal{N'}_{-i}} ) \mathbb{E} \{ \tau \bg^{[k]} + \zeta | \ro_i , \ro_{-i}\}.
\end{equation}}
By expanding the summation in (\ref{eq: value linearized}) and considering $c_3$, one can obtain that
\vspace{-4mm}
{\small\begin{align} \label{eq: expanded vk+1}
   & V^{[k+1]}(\ro_i = \ro) = 
    \sum_{\ro_{k+1} \in \Omega} \!\!\! p(\ro_{k+1}) \!\!\!\! \!\!\! \sum_{\ro_{-i}\in \Omega^{k-1}} \!\!\!\!\! p(\ro_{-i} = [\ro_j]_{j \in \mathcal{N}_{-i}} ) \mathbb{E} \{\tau \bg^{[k]} + \zeta | \ro_i , \ro_{-i}, \ro_{k+1} \}. 
\end{align}}
By applying the law of iterated expectations on eq. (\ref{eq: expanded vk+1}) we can express $V^{[k+1]}(\ro_i = \ro)$ as
\vspace{-4mm}
{\small\begin{align}
   & V^{[k+1]}(\ro_i = \ro) = 
   & \!\!\!\! \!\!\! \sum_{\ro_{-i}\in \Omega^{k-1}} \!\!\!\!\! p(\ro_{-i} = [\ro_j]_{j \in \mathcal{N}_{-i}} ) \mathbb{E} \{\tau \bg^{[k]} + \zeta | \ro_i , \ro_{-i} \}, \notag
\end{align}}
which allows us to directly imply
\vspace{-4mm}
{\small\begin{align} \label{eq: relationship of vk+1 and vk}
    V^{[k+1]}(\ro_i = \ro) = \tau V^{[k]}(\ro_i = \ro) + \zeta.
\end{align}}
Upon the correctness of the statement $P(k)$, and for the constant values of $\tau, \zeta \in \mathbb{R}$ one can also state that
\vspace{-6mm}
{\small \begin{align} \label{eq: induction: hypothesis prime}
& 
       \text{ If } V^{[2]}(\ro') > V^{[2]}(\ro) \implies \tau V^{[k]}(\ro') + \zeta > \tau V^{[k]}(\ro) + \zeta, ~ \forall \ro \in \Omega.
\end{align}}
Eq. (\ref{eq: relationship of vk+1 and vk}) together with eq. (\ref{eq: induction: hypothesis prime}) are sufficient to establish the proof of the induction step, which in turn concludes the proof. 
\end{proof}
\end{document}